\documentclass[draftcls,11pt,onecolumn]{IEEEtran}

\usepackage{times,amsmath,color,amssymb,graphicx,epsfig,cite,psfrag,enumerate}
\usepackage{subfigure,multirow,bm}
\newtheorem{Remark}{\it Remark}[section]
\newtheorem{Theorem}{\it Theorem}[section]
\newtheorem{Proposition}{\it Proposition}[section]
\newtheorem{Lemma}{\it Lemma}[section]
\newtheorem{Corollary}{\it Corollary}[section]
\newtheorem{Definition}{\it Definition}[section]

\begin{document}
%

\title{Optimal Power Allocation for Outage Probability Minimization in Fading Channels with Energy Harvesting Constraints}

\author{\IEEEauthorblockN{Chuan Huang,~\IEEEmembership{Member,~IEEE},~Rui Zhang,~\IEEEmembership{Member,~IEEE},\\Shuguang
Cui},~\IEEEmembership{Senior Member,~IEEE}

\thanks{Chuan Huang and Shuguang Cui are with the Department of Electrical
and Computer Engineering, Texas A\&M University, College Station,
TX, 77843. Emails: \{huangch, cui\}@tamu.edu.

Rui Zhang is with the Department of Electrical and Computer Engineering, National University of Singapore, Singapore 117576. Email: elezhang@nus.edu.sg.
}}
%

\maketitle
\begin{abstract}
This paper studies the optimal power allocation for outage probability minimization in point-to-point fading channels with the energy-harvesting constraints and channel distribution information (CDI) at the transmitter. Both the cases with non-causal and causal energy state information (ESI) are considered, which correspond to the energy-harvesting (EH) rates being known and unknown prior to the transmissions, respectively. For the non-causal ESI case, the average outage probability minimization problem over a finite horizon of $N$ EH periods is shown to be non-convex for a large class of practical fading channels. However, the globally optimal ``offline'' power allocation is obtained by a forward search algorithm with at most $N$ one-dimensional searches, and the optimal power profile is shown to be non-decreasing over time and have an interesting ``save-then-transmit'' structure. In particular, for the special case of $N=1$, our result revisits the classic outage capacity for fading channels with uniform power allocation. Moreover, for the case with causal ESI, we propose both the optimal and suboptimal ``online'' power allocation algorithms, by applying the technique of dynamic programming and exploring the structure of optimal offline solutions, respectively.
\end{abstract}




\section{Introduction}

Traditional energy-constrained wireless communication systems, e.g., wireless sensor networks, are equipped with fixed energy supply devices such as batteries, which have limited operation time. When a sufficient number of sensors run out of battery, the whole network may fail. Therefore, for applications in which replacing batteries is inconvenient or even impossible, e.g., in hostile or toxic environment, energy harvesting (EH) becomes an alternative solution to provide almost unlimited, easy, and safe energy supplies for wireless networks. However, compared with conventional time-invariant energy sources, energy replenished by harvesters is usually intermittent over time, e.g., with energy fluctuation caused by time-dependent solar and wind patterns. As a result, wireless devices powered by renewable energy sources are subject to explicit EH constraints over time, i.e., the total energy consumed up to a given time must be less than the energy harvested by that time.

Wireless communication with EH nodes has recently drawn significant research interests (see, e.g., \cite{sharma,yang,ozel,rui,chuan,ulukus,rajesh}). In \cite{ulukus,rui}, the authors investigated the throughput maximization problem over a finite number of transmission blocks for both the cases with a deterministic EH model (for which the energy amount and arrival time are assumed to be known \emph{a priori} at the transmitter) and a random EH model (for which only the statistics of the EH process is known). In \cite{chuan}, the authors studied the throughput maximization problem for the three-node relay channel with the deterministic EH model and the decode-and-forward (DF) relaying scheme. From an information-theoretic viewpoint, the authors in \cite{ozel,rajesh} derived the capacities for the additive white Gaussian noise (AWGN) and fading channels with random EH sources, respectively, where it was shown that assuming an infinite energy storage volume, the same channel capacity can be achieved as the cases with conventional constant power supplies.

The aforementioned works about the EH systems in fading channels investigate the maximum achievable average throughput over a long term. It is worth noting that for conventional cases with constant power supplies, the information-theoretic limits of fading channels have been thoroughly studied in the literature (e.g., see \cite{biglieri} and references therein) for both the cases without and with a transmission delay constraint, termed as the \emph{ergodic capacity} and the \emph{outage capacity}, respectively. Specifically, with the perfect channel state information (CSI) known at both the transmitter and the receiver, the ergodic capacity is achieved by a random Gaussian codebook with infinite-length codewords, which span over all different fading states \cite{goldsmith}. To maximize the ergodic capacity, it was shown in \cite{goldsmith} that the so-called ``water-filling'' power allocation is optimal. Under the same CSI assumption, the outage probability minimization problem for block-fading channels with delay-constrained transmission has been studied in \cite{caire}, where the optimal power allocation is shown to have a ``truncated channel inversion'' structure. With perfect CSI known at the receiver but no CSI is known at the transmitter, the ergodic capacity of block-fading channels is achieved by a constant transmit power allocation over all the fading states \cite{caire,telater}. Under the above CSI assumption, the outage capacity is defined as the maximum transmission rate subject to a given outage probability constraint \cite{shamai,tse}, where the transmit power is assumed to be constant over all transmission blocks.

In this paper, we consider the transmissions of delay-constrained traffic at a constant rate over block-fading channels, where the CSI is perfectly known at the receiver but only the channel distribution information (CDI) is known at the transmitter. Under such assumptions, we investigate the average (over time) outage probability minimization problem over a finite horizon of $N$ EH periods, each of which contains $M$ communication blocks, with the EH rate during each EH period being constant, denoted by $Q_i$, $1 \leq i \leq N$. In particular, we consider the following two types of EH models:
\begin{enumerate}
  \item Non-causal energy state information (ESI): All $Q_i$'s are assumed to be known at the transmitter prior to the transmissions. In this case, we are interested in developing the optimal \emph{offline} power allocation to minimize the average outage probability.
  \item Causal ESI: At the $i$-th EH period, $i=1,2,\cdots,N$, only $Q_1,\cdots,Q_i$ are known at the transmitter. For this case, we are interested in developing various \emph{online} power allocation algorithms based on the statistical properties of the EH rates $Q_i$'s.
\end{enumerate}
Note that if $N=1$, both the cases with non-causal and causal ESI become identical. As such, with $N=1$, we first examine the properties of the outage probability as a function of the transmit power for a large class of practical fading channel distributions, and show that the outage probability usually behaves as a ``concave-convex'' function. Due to this fact, it turns out that the outage probability minimization problem with power control subject to the EH constraints is in general non-convex. For the special case of $N=1$, we show that this problem degrades to the classic outage minimization problem with an average power constraint over a finite horizon of $M$ communication blocks. Although this simplified problem is still non-convex, we obtain its globally optimal power allocation with only a one-dimensional search by exploiting the properties of the outage probability function. The optimal solution is shown to have a threshold-based structure: When the average power is larger than a threshold determined by the transmission rate and the CDI of the fading channel, uniform power allocation is optimal; otherwise, the optimal power allocation corresponds to an ``on-off'' transmission. Interestingly, this result revisits the classic definition of the outage capacity for fading channels \cite{shamai} with uniform power allocation only, by revealing that a non-uniform (i.e., on-off) power allocation can yield a smaller outage probability in the low-power regime (or equivalently in the high-outage regime). Furthermore, we propose a suboptimal power allocation scheme based on the optimal threshold and on-off structures, which avoids the exhaustive search operation in the optimal solution and is shown to be asymptotically optimal as $M$ goes to infinity.

For the more general case of $N>1$, we first look at the model with non-causal ESI, where we develop a forward search algorithm to compute the globally optimal offline power allocation for outage minimization, with at most $N$ one-dimensional searches based on the results obtained for the case of $N=1$. It is shown that the obtained optimal power profile is non-decreasing over time and has an interesting ``save-then-transmit'' structure similarly to that reported in \cite{luo,xu} under different problem setups. Furthermore, a suboptimal low-complexity power allocation algorithm is proposed to avoid the exhaustive search operations in the optimal algorithm, similar to the case of $N=1$.

For the causal ESI model with $N>1$, we study the optimal online power allocation problem to minimize the average outage probability over a finite horizon, which is formulated as a Markov decision process (MDP) problem, and obtain the solution by applying dynamic programming and exploring the results obtained for the case of $N=1$. We also propose a suboptimal power allocation scheme, which achieves a more flexible performance-versus-complexity tradeoff than the MDP-based solution.

The remainder of this paper is organized as follows. Section II presents the system model and formulates the outage probability minimization problem. Section III presents some important properties of the outage probability function and applies them to solve the outage probability minimization problem for the case with $N=1$. Section IV extends the results to the more general case of $N>1$ with non-causal ESI. Section V addresses the power allocation problem for the causal ESI case with $N>1$. Numerical results are presented in Section VI to validate the proposed algorithms. Finally, Section VII concludes the paper.

\section{System Model and Problem Formulation}

\subsection{System Model}

In practical wireless systems, the duration of a communication block is usually on the order of millisecond, while the EH process evolves at a much slower speed, e.g., solar and wind power typically remains constant over windows of seconds. In other words, the coherence time of the EH process is usually much larger than the channel coherence time. As such, it is reasonable to assume that the EH process keeps constant over a sufficient number of communication blocks, during which the channel may change from block to block. In this paper, we consider the transmission over $N$ EH periods, each of which contains $M$ communication blocks, with a unit block length. During each EH period, the EH rate level is assumed to be constant over the constituting $M$ communication blocks, denoted as $Q_i$, $i=1,\cdots,N$. For the time being, we assume that $N$ and $M$ are finite integers, and will address the case of $M \rightarrow \infty$\footnote{This scenario can be regarded as an approximation for the case with very large $M$, e.g., the wind usually changes in several seconds, during which thousands of communication blocks have been sent.} in Sections III and IV. We also assume that the battery capacity to store the harvested energy is infinite, and the consumed energy at the source other than the transmission energy is relatively small and thus negligible.

We consider a block-fading channel, and the input-output relationship of the $(i,j)$-th\footnote{We denote the index number of the $j$-th communication block in the $i$-th EH period as $(i,j)$.} communication block is given by
\begin{align}
y_{i,j}  = h_{i,j} \sqrt{P_{i,j}} x_{i,j}  + n_{i,j} ,~i=1,\cdots,N,~j=1,\cdots,M,
\end{align}
where $y_{i,j}$ is the channel output, $x_{i,j}$ is the corresponding channel input with zero mean and unit average power, $h_{i,j}$ is the complex fading channel gain, $P_{i,j}$ is the transmit power, and $n_{i,j}$ is the independent and identically distributed (i.i.d.) circularly symmetric complex Gaussian (CSCG) noise with zero mean and unit variance. In particular, the complex channel gain $h_{i,j}$ is assumed to be an i.i.d. random variable across the communication blocks, with zero mean and unit variance. Here, $h_{i,j}$ is unknown to the transmitter, but perfectly known at the receiver.

For each communication block, the corresponding instantaneous mutual information of the channel assuming the optimal Gaussian codebook for the transmitted signals is given by
\begin{align}
I_{i,j}(h_{i,j} , P_{i,j}) = \log \left( 1 + |h_{i,j}|^2 P_{i,j} \right),
\end{align}
where the logarithm is assumed to have base 2. In this paper, we assume that all $NM$ communication blocks are transmitting with the same constant rate $R$; then, the outage probability at the $(i,j)$-th communication block is defined as
\begin{align}
\mathcal{F}(P_{i,j}) = \Pr \left\{ I_{i,j}(h_{i,j} , P_{i,j}) < R \right\} = \Pr \left\{ |h_{i,j}|^2 < \frac{ 2^R -1 }{P_{i,j}} \right\}.
\end{align}
Note that $\mathcal{F}(P_{i,j})$ is the outage probability function in terms of the transmit power $P_{i,j}$, the probability distribution of the fading channel gain $h_{i,j}$, and the transmission rate $R$. Moreover, we assume that $\mathcal{F}(\cdot)$ is strictly decreasing over its domain of $P_{i,j} \geq 0$.

\subsection{Problem Formulation}

In this subsection, we formulate the average outage probability minimization problem over finite horizon for both the cases with non-causal and causal ESI.

\subsubsection{Non-causal ESI}

For this case, we assume that the EH rate levels of all the $N$ EH periods are known prior to the transmissions. Hence, the transmit power available for each communication block is limited by the following EH constraints:
\begin{align} \label{source_energy-const}
\sum_{i=1}^n \sum_{j=1}^m P_{i,j}  \leq M \sum_{i=1}^{n-1} Q_i + m Q_n,~n=1,\cdots,N,~m=1,\cdots,M,
\end{align}
where the effect of block length (connecting power to energy) is normalized out by the unit-length assumption. Then, the average outage probability minimization problem over the finite horizon of $NM$ communication blocks is formulated as
\begin{align}
\text{(P1)}&~~\min_{\left\{P_{i,j} \right\}}  ~ \frac{1}{N M} \sum_{i=1}^N \sum_{j=1}^M \mathcal{F} (P_{i,j}) \label{FD_model1} \\
\text{s.t.}  &~~ (\ref{source_energy-const}),~P_{i,j} \geq 0,~1 \leq i \leq N,~1 \leq j \leq M. \label{FD_model2}
\end{align}
Since the constraints in (\ref{FD_model2}) are all linear, the convexity of Problem (P1) is determined by that of the function $\mathcal{F}(\cdot)$, which will be discussed later in Section III. Here, we show some properties of the optimal solution with an arbitrary $\mathcal{F}(\cdot)$, which are summarized as the following two remarks.

\begin{Remark}
First, the optimal solution of Problem (P1) may not be unique. This is due to the fact that if the transmit power values over two arbitrary blocks are decreasing over time, we can switch them without violating the EH constraints and still achieve the same objective value. Thus, without loss of generality, in this paper we are only interested in obtaining the optimal power allocation of Problem (P1) that is non-decreasing over time, similarly to \cite{yang,chuan}.
\end{Remark}

\begin{Remark}
Second, it is worth noting that the last constraint in (\ref{source_energy-const}) is always satisfied with equality by the optimal solution of Problem (P1). This is due to the fact that any residue power after the last communication block should be used in the last block to yield a lower average outage probability, since the outage probability function $\mathcal{F}(\cdot)$ is assumed to be strictly decreasing.
\end{Remark}

\subsubsection{Causal ESI}

For this case, we assume that only the past and current EH rate levels are known at the transmitter. Similar to \cite{rui}, we assume that the EH process $\{Q_i\}$ is a first-order stationary Markov process over $i$, and its distribution $P(Q_{i+1} | Q_i)$ is known at the transmitter. Then, the battery state $B_{i,j} \geq 0 $ at the beginning of the $(i,j)$-th communication block is given by
\begin{align} \label{battery_dynamic}
B_{i,j+1} = B_{i,j} + Q_i - P_{i,j},
\end{align}
where $B_{i,j+1}$ denotes the battery energy level at the beginning of the $((i,j)+1)$-th\footnote{$(i,j)+1$ is defined as the index of the next communication block after the $(i,j)$-th one.} communication block. Thus, it is easy to check that $\{ B_{i,j} \}$ also follows a first-order Markov model. For convenience, we assume zero initial energy storage, i.e., $B_{1,1}=0$.

Therefore, for the $n$-th EH period, the minimum average outage probability with given $Q_n$ is obtained by solving the following problem
\begin{align} \label{problem_opt_online}
(\text{P2}.n)~~&\mathcal{T}_n ^*=  \min_{ \left\{ P_{n,1},\cdots,P_{n,M} \right\} } \sum_{j=1}^M \mathcal{F} (P_{n,j}) + \mathbb{E} \left[ \sum_{i=n+1}^N   \left. \sum_{j=1}^M \mathcal{F} (P_{i,j}) \right| Q_n \right] \\
\text{s. t.}~~& 0 \leq P_{i,j} \leq B_{i,j} + Q_i,~n \leq i \leq N,~1 \leq j \leq M.
\end{align}
It is worth noting that only $Q_n$ is known to the transmitter, and $Q_{n+1} ,\cdots,Q_{N}$ are random variables. The expectation in Problem (P2.$n$) is taken over $\left\{  Q_{n+1},\cdots,Q_N \right\}$; and $P_{i,j}$, $n+1 \leq i \leq N$ and $1 \leq j \leq M$, is a function of $\left\{  Q_{n+1},\cdots,Q_i \right\}$ and $\left\{ P_{n,1},\cdots,P_{n,M} \right\}$. For convenience, we denote the group of Problems (P2.$n$), $n=1,\cdots,N$, as Problem (P2). It is easy to verify that Problem (P2) is a MDP problem, for which the optimal solution will be studied in Section V.

\section{Optimal Power Allocation for the Case of $N=1$}

In this section, we first present some important properties for the outage probability function $\mathcal{F}(\cdot)$, and then apply them to derive an optimal power allocation for Problems (P1) and (P2) for the special case of $N=1$, as well as a suboptimal power allocation scheme with lower complexity. Since $N=1$, the obtained power allocation algorithms apply to both the cases with non-causal and causal ESI.

\subsection{Properties of Outage Probability Function}

In this subsection, we first show some interesting properties of the outage probability function, based on which we define two types of outage probability functions for different fading channel distributions. Without loss of generality, we adopt Weibull fading \cite{sagias} as an example, from which we later draw a general result for other fading distributions. For convenience, we omit the subscript $(i,j)$ for $h_{i,j}$ and $P_{i,j}$ in this subsection.

With Weibull fading, the complex channel coefficient $h$ can be written as \cite{sagias}
\begin{align}
h = \left( X+ jY  \right)^{2 / \beta}, \label{distribution_weill_new}
\end{align}
where $X$ and $Y$ are i.i.d. Gaussian random variables with zero mean and identical variance satisfying $\mathbb{E} \left( |h|^2 \right) = 1$, and $\beta$ is a parameter controlling the severity or the diversity of the channel fading \cite{tse}, with $\beta >0$. It is observed that when $\beta=2$, Weibull fading degrades to the well-known Rayleigh fading. The probability density function (PDF) of $|h|$ is given as $f(r) = \beta r ^{\beta -1} \exp \left( - r^{\beta} \right)$, with $r=|h| \geq 0$, and thus the corresponding outage probability function is given as \cite{sagias}
\begin{align} \label{outage_weibull}
\mathcal{F} ( P)= 1 - \exp \left(  - \left( \frac{2^R-1}{P} \right)^{\beta/2} \right).
\end{align}

In Fig. \ref{outage_curve}, we plot the outage probability function versus transmit power for Weibull fading. It is observed and also can be verified (see Preposition \ref{conv_concav_prop} below) that for all cases with different values of $\beta$, the outage probability function is non-convex. Recall that $\beta$ is an indicator of the ``diversity order'' \cite{tse} of the Weibull fading channel (larger values of $\beta$ imply higher diversity orders). Thus, the Weibull fading model is quite general for modelling practical fading channels with different degrees of diversity.

Next, we obtain the following result on the convexity of the outage probability function with Weibull fading.

\begin{figure}[!t]
\centering
\includegraphics[width=.6 \linewidth]{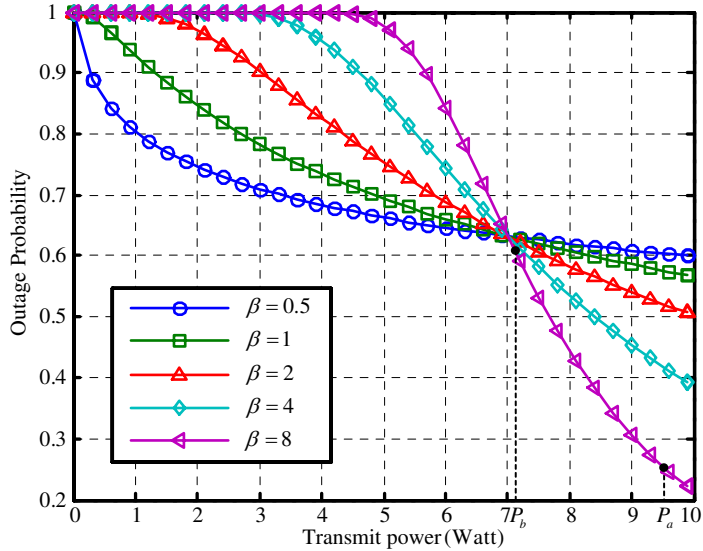}
\caption{Outage probability versus transmit power for Weibull fading with different fading parameters $\beta$, with $R=3$ bits/sec/Hz.} \label{outage_curve}
\end{figure}

\begin{Proposition} \label{conv_concav_prop}
The outage probability function given in (\ref{outage_weibull}) is concave over $P \in \left[0, P_b \right]$ and convex over $P \in \left[P_b , + \infty \right)$, where
\begin{align} \label{Pb_point}
P_b= \left( \frac{ \beta /2 }{ \beta / 2 +1} \right)^{\beta / 2} \left( 2^R-1 \right).
\end{align}
\end{Proposition}
\begin{proof}
By obtaining the second-order derivative of $\mathcal{F} (P)$ in (\ref{outage_weibull}), it follows that
\begin{align*}
\mathcal{F}''(P) &= -\frac{\beta}{2} \left( 2^R-1 \right)^{\beta/2} P^{-\beta/2 -2} \exp\left(- \left( \frac{2^R -1}{P} \right)^{\beta / 2}\right) \left( \frac{\beta}{2} \left( 2^R-1 \right)^{\beta/2} P^{- \beta/2}  - (\frac{\beta}{2} + 1) \right).
\end{align*}
We thus have $\mathcal{F}''\left(P \right) \leq 0$ over $0< P \leq P_b$ and $\mathcal{F}''(P) > 0$ over $P > P_b$, which means that $\mathcal{F}(P)$ is concave when $P \geq P_b$ and convex when $0< P \leq P_b$. Thus, the proposition is proved.
\end{proof}

Besides $P_b$, we now give another important parameter for the outage probability function as follows.

\begin{Proposition} \label{time_sharing_point}
There exists a value $P_a$, with $P_a > P_b$, such that all the points $\left( P , \mathcal{F} \left( P \right) \right)$ on the outage probability curve given in (\ref{outage_weibull}) are above the straight line passing through the two points $(0,1)$ and $\left( P_a, \mathcal{F} \left( P_a \right) \right)$, i.e.,
\begin{align} \label{time_sharing_point_formula}
\mathcal{F} (P) \geq \frac{\mathcal{F}\left( P_a \right) - 1}{ P_a} P + 1, ~\forall P \geq 0.
\end{align}
\end{Proposition}
\begin{proof}
See Appendix \ref{proof_PA}.
\end{proof}

\begin{Remark}
Beyond the existence of point $P_a$, we can further show that the point satisfying this property is unique. This claim can be validated by contradiction. Suppose that this property holds at two different points $P_a^1$ and $P_a^2$. It is then easy to check that for an arbitrary $P_0 \in [P_a^1,P_a^2]$, this property is also satisfied, which means that the function $\mathcal{F}(P)$ over the range $[P_a^1,P_a^2]$ is actually a line segment connecting the two points $(P_a^1,\mathcal{F}(P_a^1))$ and $(P_a^2,\mathcal{F}(P_a^2))$. It is thus true that the first-order derivative of function $\mathcal{F}(P)$ is a constant over $[P_a^1,P_a^2]$, while its second-order derivative is zero over this regime. However, the above cannot be true in general for a practical continuous fading distribution (see, e.g., (\ref{distribution_weill_new})). Thus, the presumption cannot be true and the uniqueness of $P_a$ is verified.
\end{Remark}

In general, it is difficult to obtain the closed-form expression for $P_a$. However, we can obtain the value of $P_a$ by a simple bisection search summarized as Algorithm \ref{search_time_sharing} in Table \ref{search_time_sharing}, for which the main ideas are given as follows: For a given point $\widetilde{P}_a$, denote the function of the straight line connecting points $(0,1)$ and $\left( \widetilde{P}_a, \mathcal{F} \left( \widetilde{P}_a \right)  \right)$ as $g(\widetilde{P}_a;x)$, which is defined as
\begin{align}
g(\widetilde{P}_a;x) = \frac{\mathcal{F}( \widetilde{P}_a)  - 1}{\widetilde{P}_a} x +1.
\end{align}
If $\widetilde{P}_a < P_a$, it is observed that there exists a $\delta>0$ such that $g(\widetilde{P}_a;x) < \mathcal{F}(x)$ for $x \in ( \widetilde{P}_a -\delta, \widetilde{P}_a )$ and $g( \widetilde{P}_a ;x) > \mathcal{F}(x)$ for $x \in ( \widetilde{P}_a , \widetilde{P}_a + \delta)$; otherwise if $\widetilde{P}_a > P_a$, it is observed that there exists an $\delta>0$ such that $g(\widetilde{P}_a ;x) > \mathcal{F}(x)$ for $x \in (\widetilde{P}_a -\delta, \widetilde{P}_a )$ and $g(\widetilde{P}_a ;x) < \mathcal{F}(x)$ for $x \in (\widetilde{P}_a, \widetilde{P}_a + \delta)$. According to the above property, we design the following algorithm to search for $P_a$, which returns an $ x \in \left( P_a-\epsilon, P_a+\epsilon \right) $ with a prescribed tolerable error $0< \epsilon \ll 1$.

\begin{table}[ht]
\begin{center}
\caption{Algorithm \ref{search_time_sharing}: Bisection search for $P_a$ defined in Proposition \ref{time_sharing_point}.}
\hrule
\vspace{0.3cm}
\begin{itemize}
\item Initialize $a_{\text{low}} =0$, $a_{\text{up}}$ (any sufficiently large positive number), and the error tolerance $\epsilon$;
\item While $a_{\text{up}}-a_{\text{low}} > \epsilon$, repeat the following steps
\begin{enumerate}
\item Let $P_a = \frac{a_{\text{up}} + a_{\text{low}}}{2}$
\item If $g(P_a;P_a-\epsilon) < \mathcal{F} (P_a-\epsilon)$ and $g(P_a;P_a+\epsilon) > \mathcal{F}(P_a+\epsilon)$, let $a_{\text{low}} = P_a$;
\item Else let $a_{\text{up}} = P_a$.
\end{enumerate}
\item Algorithm ends.
\end{itemize}
\vspace{0.2cm} \hrule \label{search_time_sharing} \end{center}
\end{table}

In Fig. \ref{outage_curve}, for the case of $\beta = 8$, we show the two points $P_a$ and $P_b$. For other types of fading channel distributions, e.g., Rician fading and Nakagami fading, it can be shown that their outage probability functions have similar properties to the above for Weibull fading; however, unlike Weibull fading, $P_b$ may no longer be expressible in a closed-form as that in (8), and can only be computed numerically. Moreover, for certain fading distributions, it is possible that the outage probability function is convex over the entire range of $P \geq 0$, which can be regarded as a special case with $P_a=P_b=0$. With the above observations, we categorize the outage probability functions for fading channel distributions in the two following types.

\begin{Definition}
An outage probability function $\mathcal{F}(x)$, $x \geq 0$, is said to be \emph{Type $\bm{A}$}, if $\mathcal{F}(x)$ is convex over $x \geq 0$.
\end{Definition}

\begin{Definition}
An outage probability function $\mathcal{F}(x)$, $x \geq 0$, is said to be \emph{Type $\bm{B}$}, if there exist two unique points $0 < P_b \leq P_a$, such that:
\begin{enumerate}

\item $\mathcal{F}(x)$ is concave over $[0,P_b]$ and convex over $(P_b, \infty)$;

\item $\mathcal{F}(x)$ is always above the line $g(x)$, passing through the two points $(0,1)$ and $(P_a, \mathcal{F}(P_a))$, i.e., $\mathcal{F}(x) \geq g(x),~x\geq 0$.

\end{enumerate}
\end{Definition}

Here, a Type B function is with a ``concave-convex'' shape, while a Type A function is a special case of Type B functions with $P_a=P_b=0$. It can be verified that Weibull fading, Rician fading, Nakagami fading, and double Rayleigh fading in general lead to Type B outage probability functions. However, we need to point out that an arbitrary outage probability function can be neither Type A nor Type B, which is out of the scope of this paper.

\begin{Remark}
It is worth noting that if $\mathcal{F}(\cdot)$ is with Type A, Problem (P1) is convex. Thus, existing power allocation algorithms in \cite{yang,rui,chuan} can be applied to solve this problem. As such, in the rest of this paper we are mainly interested in the case of Type B outage probability functions.
\end{Remark}

\subsection{Optimal Power Allocation with $N=1$}

In this subsection, we solve Problems (P1) and (P2) for the special case of $N=1$, based on the previously derived outage probability function properties. Here, since $N=1$, we use the notation $P_j$ instead of $P_{1,j}$ for convenience.

Instead of solving Problem (P1) with $N=1$ directly, we consider the following problem by removing the first $M-1$ constraints in (\ref{source_energy-const}), i.e.,
\begin{align}
\text{(P3)}&~~\min_{\left\{P_{j} \geq 0 \right\}}  ~\frac{1}{M} \sum_{j=1}^M \mathcal{F} \left(  P_{j} \right) \label{P2_FD_model1} \\
\text{s.t.}  &~~  \sum_{j=1}^M P_j \leq M Q_1 . \label{P2_FD_model2}
\end{align}
Obviously, Problem (P3) is a relaxed version of Problem (P1) with $N=1$, which provides a lower bound on the optimal value of Problem (P1). We will show next that since there exists a non-decreasing optimal solution of Problem (P3) that is guaranteed to satisfy the omitted $M-1$ constraints in (\ref{source_energy-const}), the solution for Problem (P3) is also optimal for Problem (P1) with $N=1$.

It is easy to see that Problem (P3) is also a power allocation problem to minimize the average outage probability over a $M$-block block-fading channel with an average power constraint. If $\mathcal{F}(P_j)$ is a non-convex function (e.g., Type B function), Problem (P3) is non-convex \cite{boyd}, and thus difficult to be solved by conventional convex optimization techniques. To provide a solution, we first present some structural properties of the optimal power profile of Problem (P3) as follows.

\begin{Proposition} \label{prop_one_low_n1}
For the optimal solution of Problem (P3), there is at most one strictly positive power value $P_{j}^*$ that is below $P_b$ defined in Proposition \ref{conv_concav_prop}.
\end{Proposition}
\begin{proof}
This proposition is proved by contradiction. Suppose that for the optimal power profile $P_j^{\star}$'s of Problem (P3), there exist $j_1$ and $j_2$, $1 \leq j_1,j_2 \leq M$, such that $0 \leq P_{j_1}^{\star}  \leq P_{j_2}^{\star} < P_b$. Since $\mathcal{F}(x)$ is assumed to be concave over $[0,P_b]$, define $0 < \epsilon < \min ( P_{j_1}^{\star} , P_b -P_{j_2}^{\star}) $ and it is easy to check that $\mathcal{F}(P_{j_1}^{\star}) + \mathcal{F}(P_{j_2}^{\star}) \geq \mathcal{F}(P_{j_1}^{\star} - \epsilon) + \mathcal{F}(P_{j_2}^{\star} + \epsilon)$, which means that a new power allocation, with $\widehat{P}_{j_1} = P_{j_1}^{\star} - \epsilon$ and $\widehat{P}_{j_2} = P_{j_2}^{\star} + \epsilon$, yields a lower outage probability and this contradicts with the optimality of $P_j^{\star}$'s for Problem (P3). Therefore, this proposition is proved.
\end{proof}

\begin{Proposition} \label{prop_one_high_n1}
For Problem (P3), all optimal power values $P_{j}^*$'s, that are above $P_b$ defined in Proposition \ref{conv_concav_prop}, are identical.
\end{Proposition}
\begin{proof}
The proof is similar to that of Proposition \ref{prop_one_low_n1}, by considering the convexity of the function $\mathcal{F}(x)$ over the region $\left[ P_b, +\infty \right)$.
\end{proof}

\begin{Remark} \label{opt_structure_ave_power}
Based on Propositions \ref{prop_one_low_n1} and \ref{prop_one_high_n1}, it follows that the optimal solution of Problem (P3) should have the following structure: There are at most one block assigned with power $\widehat{P}_0$, $0 < \widehat{P}_0 < P_b$, $k^*$ blocks, $0 \leq k^* \leq M$, assigned with identical power $ \frac{M Q_1 - \widehat{P}_0}{k^*} \geq P_b$, and the rest blocks with zero power. Thus, solving Problem (P3) is equivalent to finding the values for $k^*$ and $\widehat{P}_0$.
\end{Remark}

Besides $P_b$, $P_a$ defined in Proposition \ref{time_sharing_point} also plays an important role in solving Problem (P3). As shown in Appendix \ref{opt_solution}, it is always desirable to allocate the available power to be close to $P_a$ when $Q_1 < P_a$. As such, the desired $k^*$ in Remark \ref{opt_structure_ave_power} should be either $k^* = \left\lfloor  \frac{Q_1}{ P_a} M  \right\rfloor $ or $k^* = \left\lfloor  \frac{Q_1}{ P_a} M  \right\rfloor +1 $, where $\lfloor x \rfloor$ denotes the flooring operation, and $\widehat{P}_0$ can be obtained by a one-dimensional search. Therefore, we obtain the following theorem on the optimal solution of Problem (P3).

\begin{Theorem} \label{opt_power_allocation}
The optimal solution of Problem (P3) is given as follows:
\begin{enumerate}
\item If $Q_1 \geq P_a$, $P_j^* = Q_1$, $1 \leq j \leq M$;
\item If $Q_1 < P_a$,
\begin{align} \label{opt_power_allocation_low_SNR}
P_j^* = \left\{
\begin{array}{ll}
  0, & 1 \leq j \leq M-k_0 -1 \\
  \widehat{P}_0 , & j = M-k_0 \\
  \frac{ M Q_1 - \widehat{P}_0 }{k_0}  , & M-k_0 +1  \leq j \leq M
\end{array}
\right. ,
\end{align}
where $k_0$ and $\widehat{P}_0$ are given as follows:
\begin{align}
k_0 & = \left\lfloor  \frac{Q_1}{ P_a} M  \right\rfloor, \label{search_k0} \\
\widehat{P}_0 & = \arg \min_{  P \in \mathcal{P} }  \mathcal{F} (P) +  k_0 \mathcal{F} \left(  \frac{ M Q_1 - P }{k_0} \right), \label{search_Power}
\end{align}
with $\mathcal{P} =   [0,P_b) \bigcup \frac{ M Q_1  }{k_0+1}$.
\end{enumerate}
\end{Theorem}
\begin{proof}
See Appendix \ref{opt_solution}.
\end{proof}

It is worth noting that only a one-dimensional search is needed to compute the optimal power allocation for Problem (P3) when $Q_1<P_a$. As explained in Appendix \ref{opt_solution}, since we cannot claim any monotonicity results on the right hand side of (\ref{search_Power}), an exhaustive search is necessary.

\begin{Remark}
It is worth pointing out that the outage capacity results in the literature for fading channels without CSI at the transmitter are usually based on  uniform power allocation over all the communication blocks. However, Theorem \ref{opt_power_allocation} reveals that uniform power allocation can be sub-optimal for certain fading channel distributions (e.g., Type B fading) in the case of $Q_1 < P_a$, which correspond to a low-power regime or equivalently a high-outage regime.
\end{Remark}

Since the solution given in Theorem \ref{opt_power_allocation} is non-decreasing and always satisfies the first $M-1$ constraints in (\ref{source_energy-const}), we obtain the following corollary.
\begin{Corollary}
The optimal solution obtained in Theorem \ref{opt_power_allocation} for Problem (P3) is also optimal for Problem (P1) with $N=1$.
\end{Corollary}

\begin{Remark} \label{remark_when_is_better}
Theorem \ref{opt_power_allocation} shows that in the low-power regime, uniform power allocation may be non-optimal, while on-off power allocation achieves the minimum outage probability. It is easy to check that as $M$ goes to infinity, the optimal power allocation for the case of $Q_1 \leq P_a$ converges to the following binary power allocation: The source transmits with power $P_a$ over $\frac{Q_1}{P_a}$ fraction of the total blocks, and keeps silent in the rest of blocks. This is due to the following facts: 1) $\lim_{M \rightarrow \infty} \left| \lfloor \frac{Q_1}{P_a} M \rfloor - \frac{Q_1}{P_a} M \right| = 0$, which means that $k_0 \rightarrow \frac{Q_1}{P_a} M $ as $M \rightarrow \infty$; and 2) $\frac{MQ_1}{k_0+1} \rightarrow P_a $ and thus $ \frac{MQ_1 - \widehat{P}_0 }{k_0} \rightarrow P_a$ as $M \rightarrow \infty$ for any finite $\widehat{P}_0 $, which implies that to satisfy the average power constraint, we should have $\widehat{P}_0 \rightarrow 0$.
\end{Remark}

\begin{Remark}
Compared to the conventional uniform power allocation \cite{shamai}, the performance gain in the low-power regime with the optimal power allocation is expected to be more substantial as $P_a$ increases. Taking Weibull fading as an example, it is easy to see that a larger $\beta$ corresponds to a larger $P_a$. Therefore, we conclude that for the fading channels with higher diversity orders, the performance gain will become more notable, for any given $Q_1<P_a$.
\end{Remark}

\subsection{Suboptimal Power Allocation with $N=1$}

From Theorem \ref{opt_power_allocation}, it is observed that the threshold $P_a$ plays an important role in the optimal power allocation: If $Q_1 < P_a$, all the non-zero power values, except at most one that is below $P_b$, are identical and as close to $P_a$ as possible. This implies that an on-off two-level power allocation strategy may perform close to the optimal allocation. To avoid the exhaustive search in obtaining the optimal solution, we can simply allocate the power uniformly in the ``on'' state, such that we only need to determine how many blocks should be in the ``on'' state. As such, we propose the following on-off power allocation scheme for the case of $Q_1 < P_a$ as follows:
\begin{align} \label{asym_power_all_ave}
P_j = \left\{
\begin{array}{ll}
  0, & 1 \leq j \leq  M-k_0 \\
  \widetilde{P} , &  M-k_0 +1 \leq j \leq M
\end{array}
\right. ,
\end{align}
where
\begin{align}
k_0 & = \left\lfloor  \frac{Q_1}{ P_a} M  \right\rfloor,~\widetilde{P} = \frac{M Q_1 }{k_0};
\end{align}
while for the case of $Q_1 \geq P_a$, the source transmits with power $Q_1$ for all $M$ blocks, the same as Theorem \ref{opt_power_allocation}. Obviously, when $M$ goes to infinity, (\ref{asym_power_all_ave}) is asymptotically optimal according to Remark \ref{remark_when_is_better}, i.e., $\widetilde{P} \rightarrow P_a$. We thus have the following proposition.

\begin{Proposition} \label{suboptimal_prop}
The power allocation in (\ref{asym_power_all_ave}) is asymptotically optimal for Problem (P1) with $N=1$ as $M$ goes to infinity.
\end{Proposition}

\section{Offline Power Allocation for the Case of $N >1 $}

In this section, we derive the optimal and suboptimal solutions of Problem (P1) for the general case of $N > 1$ with non-causal ESI.

\subsection{Optimal Offline Power Allocation with $N>1$}

First, we give the following results on the structures of the optimal solution to Problem (P1) with $N>1$.
\begin{Proposition} \label{prop_one_low}
For the optimal solution of Problem (P1), there is at most one strictly positive power $P_{i,j}$ that is below $P_b$ defined in Proposition \ref{conv_concav_prop}.
\end{Proposition}
\begin{proof}
The proof is similar to that of Proposition \ref{prop_one_low_n1}, and thus omitted for brevity.
\end{proof}

Proposition \ref{prop_one_low} implies that if the available transmit power for any communication block is below the value $P_b$, the source should not transmit with at most one exceptional case until more energy is harvested to make the available power above $P_b$. In addition, for Problem (P1), if the optimal power values $P_{i,j}^*$ and $P_{i,j+1}^*$\footnote{Note that when $j=M$, $P_{i,j+1}$ means $P_{i+1,1}$.} for two consecutive communication blocks are both larger than $P_b$, we obtain the following properties for a non-decreasing optimal power profile.

\begin{Proposition} \label{prop_one_high}
For the non-decreasing optimal solution of Problem (P1), any two consecutive transmit power values that are both larger than $P_b$ must satisfy the following two conditions:
\begin{enumerate}
\item If the EH constraint at the $(i,j)$-th block is not achieved with equality, we have $P_{i,j}^* = P_{i,j+1}^*$;
\item From 1), we infer that if $P_{i,j}^* < P_{i,j+1}^*$, the EH constraint at the $(i,j)$-th block is achieved with equality.
\end{enumerate}
\end{Proposition}
\begin{proof}
To show the first condition, we first assume that $P_{i,j}^* < P_{i,j+1}^*$. It is easy to check that a new power profile, defined as $\widetilde{P}_{i,j} = P_{i,j}^* +\epsilon$ and $ \widetilde{P}_{i,j+1} = P_{i,j+1}^* - \epsilon $, with $0 < \epsilon <  \frac{ P_{i,j+1}^* - P_{i,j}^* }{2}$, leads to a lower average outage probability, which contradicts with the optimality presumption. Therefore, the first condition is proved. Then, the second condition could be easily proved by using the first condition and the non-decreasing property.
\end{proof}

\begin{Remark} \label{opt_structure_EH}
From Propositions \ref{prop_one_low} and \ref{prop_one_high}, we conclude that the optimal solution of Problem (P1) for the case of $N>1$ with the non-decreasing power profile must posses the following ``save-then-transmit'' structure: Initially, the transmitter keeps silent for a certain number of communication blocks; then, it (possibly) transmits with a power smaller than $P_b$ over one communication block; after that, it keeps transmitting with power larger than $P_b$, and increases power levels right after the EH period where the harvested energy is exhausted due to 2) in Proposition \ref{prop_one_high}.
\end{Remark}

\begin{figure}[!t]
\centering
\includegraphics[width=.5 \linewidth]{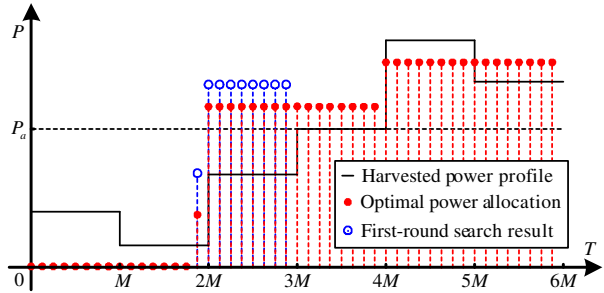}
\caption{An example of the optimal power allocation solution for Problem (P1) with $N=6$.} \label{opt_solution_example}
\end{figure}

Based on the above structure, we propose an algorithm, namely Algorithm \ref{opt_power_allocation_alg}, to compute the globally optimal solution of Problem (P1) with $N>1$, which is summarized in Table \ref{opt_power_allocation_alg}. The proposed algorithm mainly addresses the following two questions: i) when the transmission should start; and ii) how to determine the following parameters: the first positive transmission power $\widetilde{P}_0$, which might be smaller than $P_b$, and $k_0$, which is the number of communication blocks with identical power value $\widetilde{P}_1$ immediately after the first communication block with positive power value $\widetilde{P}_0$. If these two issues are solved, the remaining power allocation can be computed efficiently by Algorithm III in \cite{chuan}, since the the remaining part will be operated in the convex region of the function $\mathcal{F}(\cdot)$. It will be shown next that the above two problems are equivalent to finding an EH period with index $i_{\widetilde{t}}$ such that the result in Case 2) of Theorem \ref{opt_power_allocation} could be applied from the first to the $i_{\widetilde{t}}$-th EH periods to obtain the corresponding parameters $k_0$, $\widetilde{P}_0$, and $\widetilde{P}_1$.

Before presenting the general algorithm, we show an example in Fig. \ref{opt_solution_example} to illustrate the main ideas involved. The proposed algorithm implements a forward search from the first to the $N$-th EH period, in order to determine a particular EH period (indexed by $i_{\widetilde{t}}$) within which there is one (if any) communication block with a positive power value below $P_b$. As shown in Fig. \ref{opt_solution_example}, the following cases may occur in this search:
\begin{enumerate}
\item The EH rate values $Q_1 $ and $ Q_2$ over the 1st and 2nd EH periods are below $P_a$. It is then claimed that $i_{\widetilde{t}} > 2$ and we should continue the searching;

\item At the 3rd EH period, it is observed that $Q_3 < P_a$ and $Q_4,Q_5,Q_6 \geq P_a$. Therefore, we claim that the optimal power allocation may be no smaller than $P_a$ after the 3rd EH period,  i.e., $i_{\widetilde{t}}=3$, and thus the power allocation from the first to the 3rd EH periods can be computed similarly as Case 2) in Theorem \ref{opt_power_allocation} with equivalent average power  $\frac{1}{3} \left( Q_1 + Q_2 + Q_3 \right) $ and the number of blocks $3M$, by which we obtain $k_0$, $\widetilde{P}_0$, and $\widetilde{P}_1$.

Next, we allocate the obtained power values to the blocks starting from the $(3,M)$-th one in a backward manner to obtain a non-decreasing (in term of a forward direction) power profile, as shown by the circles in Fig. \ref{opt_solution_example}. It is easy to see that this power profile satisfies all the EH constraints up to the end of the 3rd EH period. However, since the obtained power value in the 3rd EH period is larger than the EH rate in the 4th EH period (which violates the non-decreasing power allocation for the optimal solution), we conclude that it cannot be optimal. Thus, we claim $i_{\widetilde{t}} = 4$ and use Theorem \ref{opt_power_allocation} to update $k_0$, $\widetilde{P}_0$, and $\widetilde{P}_1$ with the new average power $\frac{1}{4} \left( Q_1 + Q_2 + Q_3 + Q_4 \right) $ and the total number of blocks $4M$. Similarly, we obtain the new power profile from the first to the 4th EH periods, as shown by the dots in Fig. \ref{opt_solution_example}. Then, it is observed that the obtained power values satisfy the following two conditions: i) the power values in the 4th EH period are smaller than the average EH rate in the 5th and 6th EH periods (To comply with the optimal condition given in Proposition \ref{prop_one_high}); and ii) even if we raise the power value $\widetilde{P}_1$ by $\frac{\widetilde{P}_0 }{k_0}$ and correspondingly set $\widetilde{P}_0 = 0$, the new $\widetilde{P}_1$ is still no larger than the average EH rate in the 5th and 6th EH periods (This condition guarantees that no larger-scale search is needed). As such, we conclude that $i_{\widetilde{t}}=4$ is optimal for Problem (P1);

\item In the 5th and 6th EH periods, the harvested average power is larger than both $P_a$ and the allocated power in the 4th EH period, and thus the source should transmit at this average power. Note that in the 5th EH period, the source needs to save certain power for the 6th EH period.
\end{enumerate}

By generalizing the above procedure, we next discuss the details of Algorithm \ref{opt_power_allocation_alg} summarized in Table \ref{opt_power_allocation_alg}, which provides the optimal solution for Problem (P1) with $N>1$. For brevity, we present the details of Algorithm \ref{opt_power_allocation_alg} in Appendix \ref{deter_EH_opt}.

\begin{table}[ht]
\begin{center}
\caption{Algorithm \ref{opt_power_allocation_alg}: Compute the optimal non-decreasing power profile for Problem (P1) with $N>1$.}
\hrule
\vspace{0.3cm}
\begin{enumerate}
\renewcommand{\theenumi}{\arabic{enumi}}
\renewcommand{\theenumii}{\arabic{enumii}}
\renewcommand{\theenumiii}{\arabic{enumiii}}
\renewcommand{\theenumiv}{\arabic{enumiv}}

\makeatletter
\renewcommand{\p@enumii}{\theenumi.}
\renewcommand{\p@enumiii}{\theenumi.\theenumii.}
\renewcommand{\p@enumiv}{\theenumi.\theenumii.\theenumiii.}
\makeatother

\item Compute $P_a$ by Algorithm \ref{search_time_sharing}, and initialize $i=1$ and $temp=0$; repeat
\item Compute $i_s$ and $\widehat{P}_i$ by (\ref{deterministic_first-step1}) and (\ref{deterministic_first-step2}), and check
\begin{enumerate}
\item If $\widehat{P}_i \geq P_a$, the optimal power profile is given by (\ref{full_power}), and let $i=i_s+1$;

\item If $\widehat{P}_i < P_a$, recursively compute $i_s$ and $\widehat{P}_i$ by using (\ref{deterministic_first-step3}) and (\ref{deterministic_first-step4}), respectively. Then, check whether $t_0 \geq 1$ exists to satisfy (\ref{exist_t_0}) or not:
\begin{enumerate}
\item If $t_0$ does not exist, let $i_{s+\widetilde{t}} = N$ and compute $\widetilde{P}_0$ and $\widetilde{P}_1$ by (\ref{search_posi}) and (\ref{search_posi_p1}), respectively. Then, the optimal power allocation is given by (\ref{power_allo_case1}), and let $i=N+1$;

\item \label{recursive_compare}If $t_0$ exists, do the following operations
\begin{enumerate}

\item Let $\widetilde{t} = t_0 $, and compute $\widetilde{P}_0$ and $\widetilde{P}_1$ by (\ref{search_posi}) and (\ref{search_posi_p1}), respectively;

\item While $i_{s+ \widetilde{t}} < N $ and $ \left( (i_0,j_0) = (i_{s+ \widetilde{t}}, M) ~ \text{or}~\widetilde{P}_1 > \widehat{P}_{i+ \widetilde{t}+1 } \right)$, repeat:

    Let $\widetilde{t} = \widetilde{t} +1$, and update $\widetilde{P}_0$ and $\widetilde{P}_1$ by (\ref{search_posi}) and (\ref{search_posi_p1}), respectively.

\item Let $ temp= \widetilde{t}$;

\item \label{opt_power_EH_search_p0}While $i_{s+ \widetilde{t}} < N $ and $\widetilde{P}_0 > k_0 \left( \widehat{P}_{i+\widetilde{t}+1} -\widetilde{P}_1 \right) $, repeat:

    Let $\widetilde{t} = \widetilde{t} +1$, and update $\widetilde{P}_0$ and $\widetilde{P}_1$ by (\ref{search_posi}) and (\ref{search_posi_p1}), respectively. If conditions (\ref{conditions1}) and (\ref{conditions2}) are true, let $temp=\widetilde{t}$, and repeat this step.

\item Update $\widetilde{t} = temp$, and compute $\widetilde{P}_0$ and $\widetilde{P}_1$ by (\ref{search_posi}) and (\ref{search_posi_p1}), respectively. The optimal power profile is given as (\ref{power_allo_case1}), and let $i=i_{s+\widetilde{t}}+1$.
\end{enumerate}
\end{enumerate}
\end{enumerate}
\item Until $i > N$.
\end{enumerate}
\vspace{0.2cm} \hrule \label{opt_power_allocation_alg} \end{center}
\end{table}

\begin{Theorem} \label{deter_optimal_proof}
The solution obtained by using Algorithm \ref{opt_power_allocation_alg} is optimal for Problem (P1).
\end{Theorem}
\begin{proof}
See Appendix \ref{proof_final}.
\end{proof}

\begin{Remark}
In \cite{yang,rui}, the authors studied the power allocation problem to maximize the throughput over AWGN channels with non-causal ESI, where the throughput function is assumed to be concave. The optimal power profile for the throughput maximization problem is shown to have a continuous, non-decreasing, and piecewise-constant structure. In this paper, we consider a different problem by minimizing the outage probability over fading channels with non-causal ESI, which is non-convex due to the concave-convex shape of the outage probability function (Type B). Our results show that an on-off transmission strategy is optimal, for which the source should only transmit when the available power is sufficiently large. It is worth noting that our problem will degrade to the one considered in \cite{yang,rui} if the convex Type A outage probability function is considered.
\end{Remark}

\begin{Remark}
In Algorithm \ref{opt_power_allocation_alg}, it is observed that exhaustive searches are needed to solve the optimization problem defined in (\ref{search_posi}). Since we need to repeat the one-dimensional search in Step (\ref{recursive_compare}) of Algorithm \ref{opt_power_allocation_alg}, these search operations are executed at most $N$ times. Note that except these one-dimensional searches, the computation complexity of Algorithm \ref{opt_power_allocation_alg} is on the order of $\mathcal{O} (N^2)$ (see (\ref{deterministic_first-step1}) and (\ref{deterministic_first-step3})). Thus, the main computation burden of the proposed algorithm is due to the one-dimensional searches in (\ref{search_posi}). However, compared to searching exhaustively the optimal power allocation of Problem (P1) over the total number of $NM$ communication blocks, the computation complexity is greatly reduced with the proposed algorithm.
\end{Remark}

\begin{Remark}
By a similar argument to Remark \ref{remark_when_is_better}, the power allocation profile given in Algorithm \ref{opt_power_allocation_alg} converges to the following threshold-based transmission scheme as $M \rightarrow \infty$: At the $i$-th EH period, first use (\ref{deterministic_first-step1})-(\ref{deterministic_first-step2}) to compute the possible transmission power $\widehat{P}_i$. If $\widehat{P}_i$ is no smaller than $P_a$, we transmit with power $\widehat{P}_i$ over the $i$-th to the $i_s$-th EH periods; otherwise, keep the first $\widetilde{k} =\left\lfloor(1-\frac{\widehat{P}_i}{P_a})(i_s - i +1)M \right\rfloor$ communication blocks silent and then transmit with power $\frac{ (i_s - i +1)M }{ (i_s - i +1)M  - \widetilde{k}  } \widetilde{P}_i$ over the rest communication blocks from the $i$-th to the $i_s$-th EH periods.
\end{Remark}

\subsection{Suboptimal Offline Power Allocation with $N>1$}

In this subsection, we propose a suboptimal algorithm for Problem (P1) with lower complexity than that of Algorithm \ref{opt_power_allocation_alg}, and show that it is asymptotically optimal as $M$ goes to infinity.

Similar to the case of $N=1$ in Section III-C, we propose a suboptimal power allocation algorithm, namely Algorithm \ref{asym_opt_power}, for Problem (P1) with $N>1$, which is shown in Table  \ref{asym_opt_power}. The main idea of this algorithm is described as follows. From the first EH period, we search the index of the next possible power exhausting EH period by (\ref{deterministic_first-step1}) and (\ref{deterministic_first-step2}). If $\widehat{P}_{i} \geq P_a$, we claim that the source should transmit with its best effort, i.e., the power allocation is given as (\ref{full_power}); otherwise, an on-off transmission is adopted to guarantee that the allocated power is equal to or larger than $P_a$. The power allocation is thus given as
\begin{align} \label{power_onoff}
P_{p,q} = \left\{
\begin{array}{ll}
  0, & (i,1) \leq (p,q) \leq (i_0,j_0) \\
  \frac{(i_s - i+1) M}{k_0} \widehat{P}_{i} , & (i_0,j_0) < (p,q) \leq (i_s,M)
\end{array}
\right.,
\end{align}
where $i_0$, $j_0$, and $k_0$ are computed by (\ref{posi_power_index_i}), (\ref{posi_power_index_j}), and (\ref{posi_power_index_k0}), respectively.

Note that from the $i$-th to the $i_s$-th EH periods, the power profile is obtained by the sub-optimal solution proposed in Section III-C with $N=1$, while in Algorithm \ref{asym_opt_power}, we still use (\ref{deterministic_first-step1}) and (\ref{deterministic_first-step2}) to determine $i_s$ as in the optimal solution given in Algorithm \ref{opt_power_allocation_alg}.

\begin{table}[ht]
\begin{center}
\caption{Algorithm \ref{asym_opt_power}: suboptimal power allocation for Problem (P1) with $N>1$.}
\hrule
\vspace{0.3cm}
\begin{itemize}
\item Set $i=1$;
\item While $i \leq N$, repeat the following steps:
\begin{enumerate}
\item Compute $\widehat{P}_{i}$ and $i_s$ by using (\ref{deterministic_first-step1}) and (\ref{deterministic_first-step2});
\item If $\widehat{P}_{i} \geq P_a$, the power allocation is given by (\ref{full_power}); otherwise, the power allocation is given by (\ref{power_onoff}).

\item Let $i=i_s+1$.
\end{enumerate}
\item Algorithm ends.
\end{itemize}
\vspace{0.2cm} \hrule \label{asym_opt_power} \end{center}
\end{table}

\begin{Proposition}
The solution obtained by Algorithm \ref{asym_opt_power} is asymptotically optimal for Problem (P1) with $N>1$ as $M$ goes to infinity.
\end{Proposition}
\begin{proof}
The proof is similar to that of Proposition \ref{suboptimal_prop} for $N=1$, and thus omitted for brevity.
\end{proof}

\section{Online Power Allocation}

In this section, we consider Problem (P2) for the case with only causal ESI at the transmitter and $N>1$. First, we show that the optimal solution of Problem (P2) can be obtained by applying dynamic programming and the results in Section III for the case of $N=1$. Then, we propose a suboptimal online power allocation algorithm with lower complexity than the optimal online solution.

\subsection{Optimal Online Power Allocation}

In general, the optimization of $P_{i,j}$ for Problem (P2) with $N>1$ cannot be performed independently over each EH period, since the battery states defined in (\ref{battery_dynamic}) are coupled over time. Thus, we adopt a dynamic programming method to solve this problem, as stated in the following proposition.

\begin{Proposition}
For Problem (P2.$n$), $1 \leq n \leq N$, and the given initial states $Q_n$ and $B_{n} = B_{n,1}$, the minimum average outage probability $\mathcal{T}_n^*$ is given by $J_{n}(Q_n,B_n)$, which can be computed recursively from $J_{N} (Q_N,B_N)$ to $J_{N-1} (Q_{N-1},B_{N-1})$, until $J_{n}(Q_n,B_n)$. The sequence of optimization problems are constructed as:
\begin{align}
(\text{P}4.N) ~~~&J_{N} (Q_N,B_{N})  = \min_{ \{ P_{N,j} \} } ~\sum_{j=1}^M \mathcal{F} \left( P_{N,j} \right)  \\
\text{s.t.} ~~ &\sum_{j=1}^M P_{N,j} \leq B_{N} + M Q_N,
\end{align}
and for $1 \leq i < N$,
\begin{align}
(\text{P}4.i) ~~~ &J_{i} ( Q_i, B_{i})  =  \min_{0 \leq B_{i+1} \leq B_{i} + M Q_i} \min_{ \{ P_{i,j} \} }  ~\sum_{j=1}^M \mathcal{F} \left( P_{i,j} \right)  + \overline{J}_{i+1} (Q_i, B_{i+1}) \\
\text{s.t.} ~~& \sum_{j=1}^M P_{i,j} \leq B_{i} - B_{i+1} + M Q_i,
\end{align}
where $B_{i} = B_{i,1}$, and
\begin{align}
\overline{J}_{i+1}  (Q_i, B_{i+1})  =  \mathbb{E}_{Q_{i+1}} \left[ \left. J_{i+1} ( Q_{i+1}, B_{i+1}) \right| Q_i \right].
\end{align}
\end{Proposition}
\begin{proof}
The proof directly follows by applying the Bellman's equation \cite{bell} and using (\ref{battery_dynamic}). Note that in each Problem (P$4.i$), there should be $M$ EH constraints, i.e., $\sum_{j=1}^m P_{i,j} \leq B_{i} - B_{i+1} + m Q_i$, $1 \leq m \leq M$; however, by a similar argument as in Section III-B, we could eliminate the first $M-1$ EH constraints by finding an optimal solution $P_{i,j}$ in a non-decreasing manner.
\end{proof}

Note that Problems (P$4.i$), $1 \leq i \leq N-1$, can be solved by first applying Theorem \ref{opt_power_allocation} with fixed $B_{i+1}$, and then searching over all possible $B_{i+1}$'s, with $0 \leq B_{i+1} \leq B_{i} + M Q_i$. Therefore, the above MDP problems can be solved by dynamic programming and applying our previous results for the case of $N=1$.

\subsection{Suboptimal Online Power Allocation}

In this subsection, we propose a suboptimal online power allocation algorithm called ``$q$-period look-ahead'', which is based on the current battery state and the predicted power to be harvested for the next $q-1$ EH periods, with $q \geq 2$.

We assume that the EH process $\{ Q_i \}$ is a discreet-time first-order Markov process, and the future ESI can be predicted as
\begin{align}
\widehat{Q}_j = \mathbb{E} \left(  Q_j | Q_i \right), ~i+1 \leq j \leq \min \left\{ i+q-1, N \right\},
\end{align}
where the integer $q$ indicates the prediction window length from the current EH period. Here we only assume that the mean values of the future harvested energy are known, while the exact distribution may not be known to the transmitter, which greatly relaxes the requirements for computing the online power allocation.

Then, with the current energy profile $\widehat{Q}_i = Q_i + \frac{B_i}{M}$ and the predicted ones, $\widehat{Q}_j $, $i+1 \leq j \leq \min \left\{ i+q-1, N \right\} $, we compute the power allocation by using either the optimal or the suboptimal alogrithms, i.e., Algorithm \ref{opt_power_allocation_alg} or \ref{asym_opt_power}, and thereby adopt the power allocation for the current $i$-th EH period, i.e., from the $(i,1)$-th to the $(i,M)$-th communication blocks. Similarly, at the next EH period, we repeat the above procedure to obtain its corresponding power allocation, until the $N$-th EH period is reached. Evidently, if $q=1$, the proposed scheme becomes a ``greedy'' power allocation, i.e., at the end of each EH period, all the stored harvested energy is used up. In the next section, we will evaluate the performance of this suboptimal online algorithm.

\section{Numerical Results}

For the purpose of exposition, we consider Weibull fading with $\beta = 8$ throughout this section.

\subsection{The Case of $N=1$}

\begin{figure}[!t]
\centering
\includegraphics[width=.6 \linewidth]{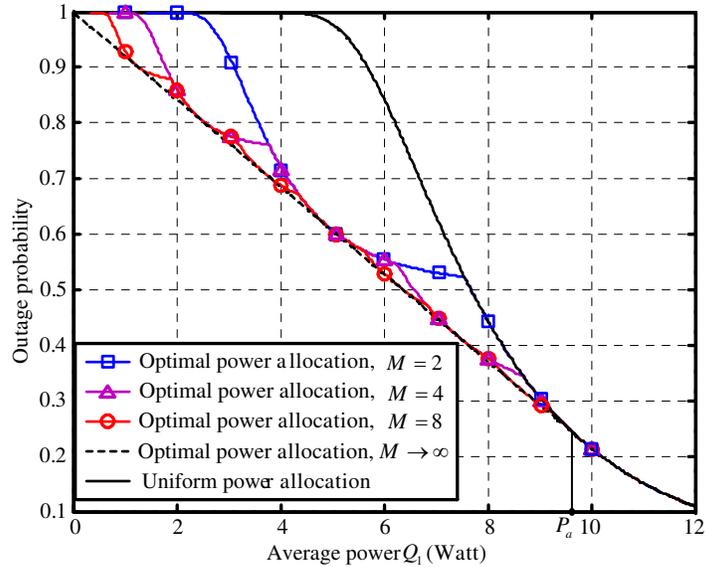}
\caption{Outage probability over the Weibull fading, with $\beta=8$ and $R=3$ bits/sec/Hz.} \label{outage_vs_block}
\end{figure}

In this subsection, we show some numerical results to validate our analysis in Section III for the case of $N=1$. In Fig. \ref{outage_vs_block}, we plot the outage probability vs. average transmit power with different numbers of communication blocks, $M$. It is observed that the outage probability with the optimal power allocation is no larger than that with uniform power allocation over the region $Q_1<P_a$. In particular, the minimum outage probability with the optimal allocation for the case of $M \rightarrow \infty$ over this region is a straight line connecting the points $[0,1]$ and $\left[ P_a,\mathcal{F}(P_a) \right]$. Moreover, as $M$ increases, the minimum outage probability converges to that with $M \rightarrow \infty$.

\begin{figure}[!t]
\centering
\includegraphics[width=.6 \linewidth]{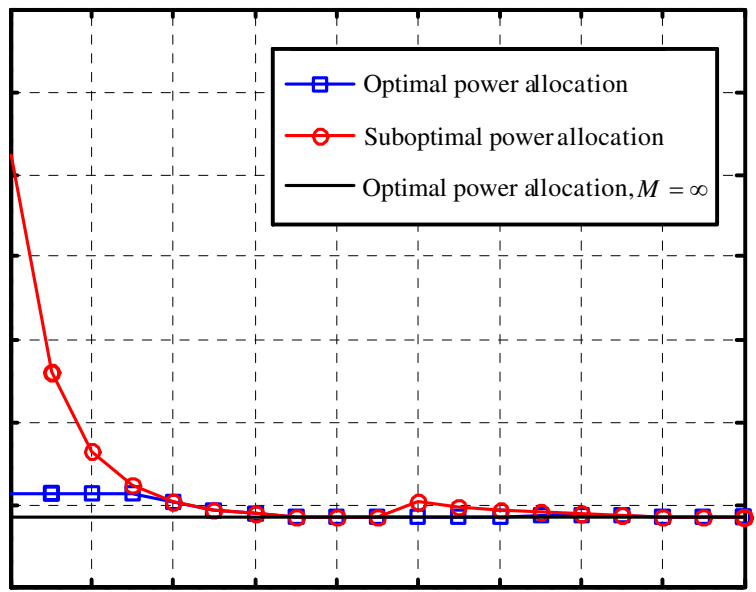}
\caption{Outage probability comparison for the optimal and suboptimal power allocation schemes, with $\beta=8$, $Q_1 = 9$, and $R=3$ bits/sec/Hz.} \label{outage_asym}
\end{figure}

In Fig. \ref{outage_asym}, we plot the average outage probability versus $M$ with both the proposed optimal and suboptimal power allocation schemes. It is observed that the outage probabilities of both schemes converge to the same value as $M$ goes to infinity, which is in accordance with Proposition \ref{suboptimal_prop}.

\subsection{The Case of $N>1$}

In this subsection, we compare the performance of the proposed offline algorithm against that of the online power allocation algorithm, for the case of $N>1$ with non-causal versus causal ESI, respectively. We assume an i.i.d. model for the EH source, where the EH rate at each EH period has three equal-probability states: 0, $P$, and $2P$, with $P$ the average EH rate value. For simplicity, we also assume that $M=1$ and $N=20$. As shown in Fig. \ref{comparison_EH_power_allocation}, we plot the outage probability of different algorithms with $P$. As a good approximation for the continuous-state MDP problem in (\ref{problem_opt_online}), we model the battery states as a finite set, where the difference between two adjacent states is set as 0.01. It is observed that the optimal online algorithm performs very close to the optimal offline algorithm. For the ``$q$-period'' look-ahead algorithm, we choose three values for $q$: $q=1$, $q=2$, and $q=N$. It is observed that a properly chosen $p$ may yield a better outage performance, e.g., in Fig. \ref{comparison_EH_power_allocation}, the $2$-period look-ahead scheme outperforms the other two cases over most of the $P$ values.

\begin{figure}[!t]
\centering
\includegraphics[width=.6 \linewidth]{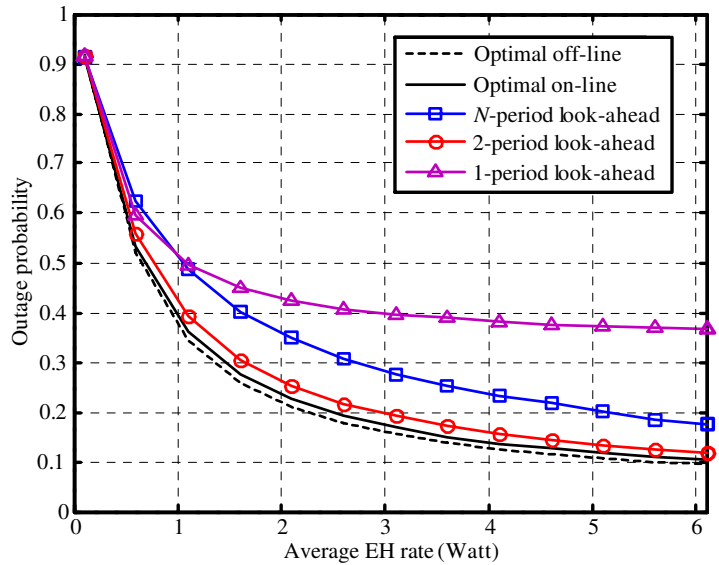}
\caption{Comparison of various power allocation algorithms, with $R=0.5$ bits/sec/Hz, $N=20$ and $M=1$.} \label{comparison_EH_power_allocation}
\end{figure}

\section{Conclusion}

In this paper, we studied the optimal power allocation to minimize the averaged outage probability in fading channels with EH constraints. We showed that for most of practical fading channels, the outage probability function is in general non-convex over the transmit power, which results in non-convex power allocation problems. We derived the globally optimal solutions for such problems by exploiting some interesting properties of the outage probability function and the causality structure of the EH constraints. It was shown that the optimal offline power allocation follows an interesting ``save-then-transmit'' protocol. For the special case of $N=1$, our results revisit the classic outage capacity problem with new observations. Furthermore, we considered the practical case with only causal ESI and proposed both optimal and suboptimal online power allocation schemes by applying dynamic programming and exploring the structure in the offline solution.

%
\IEEEpeerreviewmaketitle

\appendices

\section{Proof of Proposition \ref{time_sharing_point}} \label{proof_PA}
To prove this proposition, we draw a line passing through $(0,1)$ and any other point $(P, \mathcal{F}(P))$ on the outage probability curve defined in (\ref{outage_weibull}), and denote its slope as $\mathcal{S}(P)$, $P > 0$. Then, we can prove that the slope of this line is lower-bounded, which means that the desired point can be found corresponding to this lower bound.

To see this point, we check the slope function:
\begin{align}
\mathcal{S}(P) = \frac{\mathcal{F}(P) - 1}{P},~ P >0.
\end{align}
It is observed that $ \lim_{P \rightarrow + \infty }\mathcal{S}(P) = 0$, since the numerator is always bounded. In addition, $\lim_{P \rightarrow 0 }\mathcal{S}(P) = 0$, since $\exp \left(  - \left( \frac{2^R-1}{P} \right)^{\beta/2} \right) $ is a higher order infinitesimal of $P$ for $P \rightarrow 0$. Then, for a large enough value $A>0$, we conclude that: 1) If we further define $\mathcal{S}(0) = 0$, $\mathcal{S}(P)$ is both upper- and lower-bounded over $[0,A]$, since it is a continuous function over this region; and 2) $\mathcal{S}(P)$ is lower-bounded over $[A,+\infty)$, since it is increasing over this region  (note that $\mathcal{S}(P)$ is always non-positive). Based on the above results, this proposition is proved.

\section{Proof of Theorem \ref{opt_power_allocation}} \label{opt_solution}

First, we present the following observations for the considered Type B function.
 
\begin{Remark} \label{property_type_b_function}
For the considered concave-convex shape Type B function $\mathcal{F}(P)$, it can be proven that: 1) For $0 \leq X_1 \leq X_2 \leq X_3 \leq P_a$, the line segment (denoted as $\mathcal{L}_1(x), x\in [X_1,X_2]$) connecting the two points $(X_1,\mathcal{F}(X_1))$ and $(X_2,\mathcal{F}(X_2))$ is always above the one (denoted as $\mathcal{L}_2(x), x\in [0,X_3]$) connecting the two points $(0,1)$ and $(X_3,\mathcal{F}(X_3))$, i.e., for arbitrary $X_0$, $X_1 \leq X_0 \leq X_2$, it follows that $\mathcal{L}_1(X_0) \geq \mathcal{L}_2(X_0)$; and 2) for $ P_a \leq X_0 $ and $0 \leq X_1 \leq X_0 \leq X_2$, the line segment (denoted as $\mathcal{L}_3(x),x\in [X_1,X_2]$) connecting the two points $(X_1,\mathcal{F}(X_1))$ and $(X_2,\mathcal{F}(X_2))$ is always above the point $(X_0,\mathcal{F}(X_0))$, i.e., $\mathcal{L}_3(X_0) \geq \mathcal{F}(X_0)$.
\end{Remark}

Next, we prove the second case of $Q_1 < P_a$ by starting with the following lemma. Later we will show that the proof for the first case of $Q_1 \geq P_a$ is only a special case of this one.

\begin{Lemma} \label{ave_power_lemma}
For the following function
\begin{align}
f_{k+1}(P) = \frac{1}{k+1} \left[ \mathcal{F}(P) + k \mathcal{F} \left( \frac{MQ_1 - P}{ k} \right) \right], \label{lemma_ave_outage}
\end{align}
where $0 \leq P \leq \frac{MQ_1}{k+1}$ and $k $ is chosen from the set $ \{ 1,\cdots,M-1 \}$, its minimum value is attained at: $P = 0 $, if $ P_a \geq \frac{MQ_1}{k} $; $P=\frac{MQ_1}{k} $, if $P_a < \frac{MQ_1}{k+1} $; or search over $0 \leq P \leq \frac{MQ_1}{k+1}$, if $P_a \in \left(  \frac{MQ_1}{k+1}, \frac{MQ_1}{k} \right)$.
\end{Lemma}
\begin{proof}
We choose one particular value of $k=M-1$ as an example to prove this lemma. Consider the following $k +1 = M$ power values with the average power $Q_1$: $P_1 = P$, with $0 \leq P \leq \frac{MQ_1}{k+1}$, and $P_2 , \cdots , P_{k +1}$ all identical to the value $  \frac{MQ_1 - P}{ k}$, as shown in Fig. \ref{opt_solu_ave_power}. It is easy to check that with this power profile, its corresponding average outage probability $f_M(P) = p_a$ is given by the point $T_a = (Q_1, p_a )$, which is the intersection point between the line $x= Q_1$ and the line segment connecting the two points $(P_1, \mathcal{F}(P_1))$ and $(P_2, \mathcal{F}(P_2))$. In particular, when $P_1=P=0$, the outage probability $f_M(0) = p_a^0$ is given by the point $\hat{T}_{a} = (Q_1, p_a^0 )$, which is the intersection point between the line $x= Q_1$ and the line segment connecting the two points $(0, 1)$ and $(\hat{P}_2, \mathcal{F}(\hat{P}_2))$, where $\hat{P}_2 = \frac{MQ_1}{M-1}$. Then, by point 1) in Remark \ref{property_type_b_function}, we obtain $p_a^0 \leq p_a$, and the lemma is thus proved for $k=M-1$. 

\begin{figure}[!t]
\centering
\includegraphics[width=.6 \linewidth]{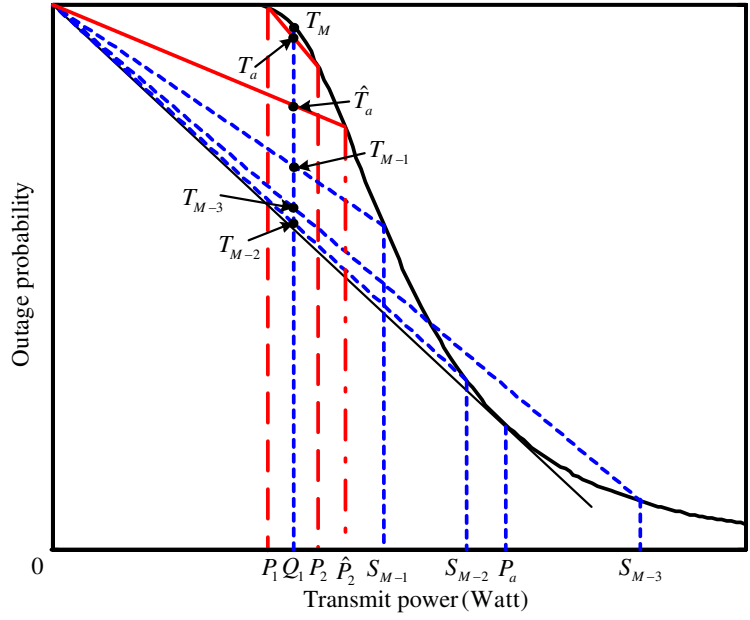}
\caption{Illustration of the optimal solution of Problem (P3).} \label{opt_solu_ave_power}
\end{figure}

Similarly, we consider the case that there are $k+1$ blocks, $1\leq k <  M-1$, and with the average power value $\frac{MQ_1}{k+1}$: We can show that i) when $P_a \geq \frac{MQ_1}{k} $, the above result is also true, and thus we claim that the function $f_{k+1}(P)$ achieves its minimum value when $P=0$; and ii) when $\frac{MQ_1}{k} > P_a$, it can be shown by point 2) in Remark \ref{property_type_b_function} that $f_{k+1}(P)$ achieves its minimum value when $P=\frac{MQ_1}{k}$. For the rest cases that $P_a \in \left(  \frac{MQ_1}{k+1} , \frac{MQ_1}{k}  \right)$, it is not sure where its minimum value is attained and thus an exhaustive search over $0 \leq P \leq \frac{MQ_1}{k+1}$ is needed.
\end{proof}

Now we are ready to prove the case of $Q_1 < P_a$ in Theorem \ref{opt_power_allocation}. From Remark \ref{opt_structure_ave_power}, it is sufficient for us to consider only the case of non-zero values for the optimal power allocation: one block with $\widetilde{P}_0$ and $k$ blocks with identical $\widetilde{P}_1 = \frac{MQ_1 - \widetilde{P}_0 }{k}$, $1 \leq k \leq M$, where $0 \leq \widetilde{P}_0 \leq \frac{MQ_1}{k+1}$. With this power profile, the average outage probability over the whole $M$ blocks is given as $p_{out} (k, \widetilde{P}_0)= \frac{1}{M} \left( k \cdot f_k ( \widetilde{P}_0 ) + M-k \right) $. In particular, denote $p_k=p_{out} (k, 0)$. By Lemma \ref{ave_power_lemma} and noting that $\frac{M-k}{M}$ is a constant for a fixed $k$, it follows that the optimal value of Problem (P3) should be either 1) one of $\left\{ p_k \right\}$, $k=1,\cdots,M$, or 2) the minimum of function $p_{out} (k_0, \widetilde{P}_0)= \frac{1}{M} \left( k_0 \cdot f_{k_0} ( \widetilde{P}_0 ) + M-k_0 \right)$, with $0 \leq \widetilde{P}_0 \leq \frac{MQ_1}{k_0+1}$. 

Then, we define a sequence of power values $\left\{ S_k \right\}$ as $S_k  = \frac{M Q_1}{ k }$, $k=1,\cdots,M$, which serves as the average power value if we equally allocate $MQ_1$ amount of power to $k$ blocks. As shown in Fig. \ref{opt_solu_ave_power}, we connect the point $(S_k, \mathcal{F}(S_k))$, $1 \leq k \leq M$, with the point $(0,1)$; this line segment then intersects with the line $x=Q_1$ at the point $T_{k}= (Q_1, p_k)$, where $p_k$ is defined in the previous paragraph. By Remark \ref{property_type_b_function}, it is easy to observe from Fig. \ref{opt_solu_ave_power} that $p_M > p_{M-1} > \cdots > p_{k_0+1}$ and $p_{k_0} < p_{k_0 -1} < \cdots < p_{1}$, where $k_0$ is given by (\ref{search_k0}). Therefore, the optimal $\widehat{P}_0$ is given by (\ref{search_Power}), and the range $\mathcal{P}$ is obtained by further considering Proposition \ref{prop_one_low_n1}. Thus, for the case of $Q_1 < P_a$, this theorem is proved. Note that power allocation for the case of $Q_1 \geq P_a$ is only a special case of that for $Q_1 < P_a$, in the sense of letting $k_0 =0$, and thus the proof is similar and omitted for brevity. Therefore, Theorem \ref{opt_power_allocation} is proved.

\color{black}

\section{Algorithm for the Optimal Power Allocation of Problem (P1) with $N>1$} \label{deter_EH_opt}

Assume that before the $i$-th EH period, the optimal power profile is obtained as $P_{p,q}^*$, $1 \leq p< i$ and $1 \leq q \leq M$. Compute the index $i_s$ corresponding to the next possible power exhausting\footnote{The power exhausting EH period means that at the end of this EH period, all EH rate is consumed.} EH period after the $i$-th EH period, and the possible constant power allocation value $\widehat{P}_{i}$ from the $i$-th to the $i_s$-th EH period, which are given by
\begin{align}
i_s & = \arg \min_{i \leq n \leq N } \left\{ \frac{Q_0 + \sum_{k=i}^n Q_k }{  n-i+1  }  \right\}, \label{deterministic_first-step1}\\
\widehat{P}_i & = \frac{ Q_0 + \sum_{k=i}^{i_s} Q_k  }{ i_s-i+1 }, \label{deterministic_first-step2}
\end{align}
where $Q_0$ is the residual power up to the $i$-th EH period with $Q_0 = \sum_{k=1}^{i-1} Q_k$ if $P_{n,m}^*=0$ for $1 \leq n <i$ and $1 \leq m \leq M$. It will be shown later that if there exists $P_{n,m}^*> 0$ for $1 \leq n <i$ and $1 \leq m \leq M$, we have $Q_0 = 0$, due to the fact that the transmitter will only save power before certain EH period and after that, it will transmit with its best effort. Then, we have the following possible cases:
\begin{enumerate}
\item Case I: If $\widehat{P}_{i} \geq P_a $, it is claimed that the source should transmit with its best effort, and the optimal power values are given as
\begin{align} \label{full_power}
P^*_{p,q} =\widehat{P}_i ,(i,1) \leq (p,q) \leq (i_s,M),
\end{align}
where $(n_1,m_1) \leq (n_2,m_2)$ is defined as one of the following two cases: i) $n_1 < n_2$; and ii) $n_1 = n_2$ and $m_1 \leq m_2$, with the equality achieved only when $n_1 = n_2$ and $m_1 = m_2$; set $i = i_s+1$, and continue the search procedure.

\item Case II: If $\widehat{P}_{i} < P_a $, we need to determine a search region (from the $i$-th to the $i_{s+\widetilde{t}}$-th EH period, $\widetilde{t} \geq 0$), where a power value below $P_b$ exists and the result in Case 2) of Theorem \ref{opt_power_allocation} can be implemented. Then, we recursively compute the index $i_{s+t}$, $t=1,2,\cdots$, of the next possible power exhausting EH period after the $i_{s+t-1}$-th EH period and the corresponding power value $\widehat{P}_{i+t}$ from the $i_{s+t-1}$-th to the $i_{s+t}$-th EH period as:
\begin{align}
i_{s+t} & = \arg \min_{i_{s+t-1}+1 \leq n \leq N } \left\{ \frac{ \sum_{k=i_{s+t-1}+1}^n Q_k }{  n - i_{s+t-1}  }  \right\}, \label{deterministic_first-step3} \\
\widehat{P}_{i+t} & = \frac{ \sum_{k=i_{s+t-1}+1}^{i_s} Q_k  }{ i_{s+t}-i_{s+t-1} }. \label{deterministic_first-step4}
\end{align}
Then, we consider the following two cases:

\begin{enumerate}
\item If there does not exist any $t_0 \geq 1$ such that
\begin{align} \label{exist_t_0}
\widehat{P}_{i+t_0 } < P_a,~\widehat{P}_{i+t_0+1} \geq P_a,
\end{align}
we obtain that $\widehat{P}_{N} < P_a$, and thus claim that $i_{s+\widetilde{t}} = N$. Then, we apply the same searching procedure as shown in Case 2) of Theorem \ref{opt_power_allocation} among the $i$-th to the $i_{s+\widetilde{t}}$-th EH periods with a total energy budget of $M \sum_{t=0}^{\widetilde{t}} \left( i_{s+t} - i_{s+t-1}  \right) \widehat{P}_{i+t}$, and compute the index $(i_0,j_0)$ corresponding to the unique communication block with positive power value less than $P_b$ as follows:
\begin{align}
i_0 & = i_{s+\widetilde{t}} - \left\lfloor \frac{k_0}{M} \right\rfloor, \label{posi_power_index_i} \\
j_0 & = M - \left( k_0 - \left\lfloor \frac{k_0}{M} \right\rfloor M \right)  , \label{posi_power_index_j}
\end{align}
where
\begin{align} \label{posi_power_index_k0}
k_0   = \left\lfloor \frac{ M \sum_{t=0}^{\widetilde{t}} \left( i_{s+t} - i_{s+t-1}  \right) \widehat{P}_{i+t} }{P_a} \right\rfloor,
\end{align}
with $i_{s-1} = i-1$. Next, the optimal power value $\widetilde{P}_0$ of the $(i_0,j_0)$-th communication block needs to be within the interval $\left[ 0, P_0 \right]$, where
\begin{align} \label{opt_alg_scale}
P_0 = \frac{( i_{s+\widetilde{t}} - i +1) M \widehat{P}_i }{ k_0 +1 } ,
\end{align}
and the power values from the $\left( (i_0,j_0)+1 \right)$-th to the $(i_{s+\widetilde{t}},M)$-th communication blocks are identical, denoted as $\widetilde{P}_1$. Due to Proposition \ref{prop_one_low}, we only need to search $\widetilde{P}_0$ over the region $\mathcal{P} = [0, P_0]/[P_b,\max(P_b,P_0))$. Then, $\widetilde{P}_0$ is given as the optimal point of the following one-dimensional outage probability minimization problem
\begin{align}
p_{\widetilde{t}} =  \min_{ P \in \mathcal{P}}  \mathcal{F} (P) + k_0  \mathcal{F} \left(  P_0 + \frac{P_0 - P}{ k_0 } \right) \label{search_posi}.
\end{align}
Furthermore, the power values from the $(i_0,j_0+1)$-th to the $(N,M)$-th communication blocks are given as
\begin{align} \label{search_posi_p1}
\widetilde{P}_1 = P_0 + \frac{P_0 - \widetilde{P}_0}{ k_0 }.
\end{align}
Therefore, the optimal power allocation from the $i$-th to the $i_{s+\widetilde{t}}$-th EH periods is given as
\begin{align} \label{power_allo_case1}
P_{p,q}^* = \left\{
      \begin{array}{ll}
        0, & (i,1) \leq (p,q) \leq (i_0,j_0)-1 \\
        \widetilde{P}_0 , & (p,q) = (i_0,j_0) \\
        \widetilde{P}_1 , & (i_0,j_0)+1 \leq (p,q) \leq (i_{s+\widetilde{t}},M)
      \end{array}
    \right.
,
\end{align}
where $(n,m)-1$ is defined as the index of the previous communication block before the $(n,m)$-th one.

\item \label{N_larger_one_dim}If such a $t_0$ exists to satisfy (\ref{exist_t_0}), we claim that the only positive optimal power value below $P_a$ may appear among the $i$-th to the $i_{s+ \widetilde{t}}$-th EH periods, with $\widetilde{t}= t_0$, whose index $(i_0,j_0)$ can be computed by using (\ref{posi_power_index_i}) and (\ref{posi_power_index_j}) with $ \widetilde{t}  = t_0$. Accordingly, $\widetilde{P}_0$ and $\widetilde{P}_1$ can also be computed by using (\ref{search_posi}) and (\ref{search_posi_p1}), respectively. With the obtained parameters, we need to further check whether they are optimal. To do this, we should check whether Steps (\ref{opt_case1}) and (\ref{opt_case2}) shown in below are satisfied or not, and update the parameter $\widetilde{t}$ for which the result in Case 2) of Theorem \ref{opt_power_allocation} can be applied (If it is updated, a larger-scale search is needed.) and the corresponding power allocation profile.

\begin{enumerate}
\item \label{opt_case1}If $(i_0,j_0) = (i_{s+\widetilde{t}},M)$ or $\widetilde{P}_1 > \widehat{P}_{i+ \widetilde{t} +1}$, we should further search over the $i$-th to the $i_{s+\widetilde{t}+1}$-th EH periods. It is noted that: (1) If the former subcase happens, the total sum EH rate up to the $i_{s+\widetilde{t}}$-th EH period is less than $P_b$, which means that we actually have not done the one-dimension search (as in Case 2) of Theorem \ref{opt_power_allocation}) from the $i$-th to the $i_{s+\widetilde{t}}$-th EH period, and thus we need to enlarge the searching scope to find the optimal solution; and (2) if the latter subcase happens, it cannot be optimal due to the fact that if $\widetilde{P}_1 > \widehat{P}_{i+\widetilde{t}+1}$, we can save power for the $(i_{s+\widetilde{t}}+1)$-th to the $i_{s+\widetilde{t}+1}$-th EH periods to decrease the average outage probability, since over the power regime $\left[ \widehat{P}_{i+\widetilde{t}+1}, \widetilde{P}_1 \right]$, the outage probability function is convex.

Then, we update $\widetilde{t} = \widetilde{t}+1$, and compute $\widetilde{P}_0$ and $\widetilde{P}_1$ again by using (\ref{opt_alg_scale}), (\ref{search_posi}), and (\ref{search_posi_p1})\footnote{Here, we do not need to update $(i_0,j_0)$, since both the previously obtained $\widetilde{P}_1$ and $\widehat{P}_{i+\widetilde{t}+1}$ can be shown to be larger than $P_a$.}. We repeatedly check whether the above conditions are satisfied or not by the newly obtained $(i_0,j_0) $, $\widetilde{P}_0$, and $\widetilde{P}_1$, until the $N$-th EH period. If the $N$-th EH period is reached, we claim that the obtained power allocation profile is optimal; otherwise, we go on to check the next Step (\ref{opt_case2}).

\item \label{opt_case2}If $(i_0,j_0) <(i_s,M)$, $\widetilde{P}_1 \leq \widehat{P}_{i+\widetilde{t}+1}$, and $\widetilde{P}_0 > k_0 (\widehat{P}_{i+\widetilde{t}+1} - \widetilde{P}_1)$, it is claimed that the obtained $(i_0,j_0)$ is the location of the last communication block with power value less than $P_b$, while its corresponding power value $P_{i_0,j_0}$ may not be optimal. Under this condition, it can be shown that the optimal power value should satisfy $P_{i_0,j_0}^* \leq P_{i_0,j_0}$. Thus, as we decrease $P_{i_0,j_0}=\widetilde{P}_0 $ by $\Delta$, $0 < \Delta \leq \widetilde{P}_0$, and correspondingly increase $\widetilde{P}_1$ by $\frac{\Delta}{k_0}$, we can possibly make $\widetilde{P}_1 + \frac{\Delta}{k_0} >  \widehat{P}_{i+\widetilde{t}+1}$ for some $\Delta$, which means that a larger-scale search (up to the $i_{s+\widetilde{t}+1}$-th EH period) may help with decreasing the outage probability. Since it can be shown that the obtained $(i_0,j_0)$ are fixed, $P_{i_0,j_0}$ can only be changed over the region $[0,\widetilde{P}_0]$.

Next, we try this larger-scale search, and check whether it is optimal. We initialize a variable $temp = \widetilde{t}$, and use it to store the index corresponding to the search with the lowest outage performance. Then, let $\widetilde{t} = \widetilde{t}+1$, and update the parameters $\widetilde{P}_0$ and $\widetilde{P}_1$ again by using (\ref{opt_alg_scale}), (\ref{search_posi}), and (\ref{search_posi_p1}). We claim that only when both of the following two conditions are satisfied, this search can improve the outage performance, and update $temp = \widetilde{t}$: (i) The EH constraint at the $(i_{s+ \widetilde{t}-1},M)$-th communication block should be satisfied by the newly obtained power profile, i.e.,
\begin{align} \label{conditions1}
\widetilde{P}_0 + \left(  (i_{s+ \widetilde{t}-1} - i_0+1 ) M - j_0 \right)\widetilde{P}_1 \leq M \sum_{i=1}^{i_{s+ \widetilde{t}-1}} Q_i.
\end{align}
This is due to the fact that in (\ref{posi_power_index_k0}), we allocate the power without checking the feasibility of this constraint; and (ii) the newly derived power profile leads to a smaller outage probability over the $i$-th to the $i_{s+\widetilde{t} }$-th EH periods, i.e.,
\begin{align} \label{conditions2}
p_{temp} + \sum_{r=temp+1}^{\widetilde{t}} \left( i_{s+r} - i_{s+r -1}  \right) M \mathcal{F} (\widehat{P}_{i+r})  >  p_{\widetilde{t}}.
\end{align}
We should repeat step (\ref{opt_case2}) until either the $N$-th EH period is reached or the above conditions (i)-(ii) are not satisfied. Then, $temp$ provides the value of $\widetilde{t}$ as needed, i.e., let $\widetilde{t} = temp$, and go to Step (\ref{opt_case3});

\item \label{opt_case3}Now, we use the obtained value of $\widetilde{t}$ to compute the optimal solution as given by (\ref{power_allo_case1}).

\end{enumerate}

Then, Step (\ref{N_larger_one_dim}) is finished. Let $i= i_{s+\widetilde{t}}+1$, and repeat the searching procedure.

\end{enumerate}

\end{enumerate}

Since there is at most one communication block with power value less than $P_b$, the searching procedure for Step (\ref{N_larger_one_dim}) needs to be implemented only once. As such, if this case has occurred, we can repeat the searching procedure in Case 1) only. It is easy to verify that the optimal solution of Problem (P1) obtained by Algorithm \ref{opt_power_allocation_alg} is non-decreasing over time, and corresponds to a ``save-then-transmit'' structure: When the source initially does not have sufficient power, i.e., much smaller than $P_a$, to support the communication rate $R$ with low outage probabilities, it keeps silent to save energy until the cumulated power become sufficiently close to $P_a$; afterwards, it starts transmitting continuously with non-decreasing power.

\section{Proof of Theorem \ref{deter_optimal_proof}} \label{proof_final}

Denote the optimal solution of Problem (P1) as $\{P_{i,j}^*\}$, and the power profile obtained by Algorithm \ref{opt_power_allocation_alg} as $\{P_{i,j}^{\star}\}$. For $\{P_{i,j}^{\star}\}$, we denote the indices of the power changing blocks as $(1,1) \leq (i_1,j_1)<\cdots<(i_n,j_n)<(i_{n+1},j_{n+1})= (N,M)$, which means that from the $(1,1)$-th to the $((i_1,j_1)-1)$-th communication block, we have $P_{i,j}^{\star} = 0$, and from the $(i_{k-1},j_{k-1})$-th, $2\leq k \leq n+1$, to the $((i_{k},j_{k})-1)$-th communication blocks, the optimal power values $P_{i,j}^{\star}$'s keep identical. To simplify the proof, we assume that $(i_2,j_2) = (i_1,j_1)+1$, and $(i_1,j_1)$ denotes the location of the one possible communication block with power value less than $P_b$; however, for the case that all positive power values are no smaller than $P_b$, this assumption will not introduce any trouble in the following proof. Moreover, by Proposition \ref{prop_one_high}, we know that the EH constraint at the $((i_k,j_k)-1)$-th, $3 \leq k\leq n+1$, communication block is satisfied with equality. Similarly, we define the indices of the power changing blocks for the power profile $\{P_{i,j}^*\}$ as $\{(i'_k,j'_k)\}$ for $0 \leq k \leq m$.

We assume that before the $(p,q)$-th communication block, the above defined two power profiles are the same, and their first difference appears at the $(p,q)$-th communication block, i.e., $P_{p,q}^* \neq P_{p,q}^{\star}$. Then, we consider the following three cases:
\begin{enumerate}
\renewcommand{\theenumi}{\arabic{enumi}}
\renewcommand{\theenumii}{\arabic{enumii}}
\renewcommand{\theenumiii}{\arabic{enumiii}}
\renewcommand{\theenumiv}{\arabic{enumiv}}

\makeatletter
\renewcommand{\p@enumii}{\theenumi.}
\renewcommand{\p@enumiii}{\theenumi.\theenumii.}
\renewcommand{\p@enumiv}{\theenumi.\theenumii.\theenumiii.}
\makeatother

\item \label{proof_opt_case1}If $(p,q) < (i_1,j_1)$, it is easy to see that $P_{p,q}^{\star}=0$ and only $P_{i_1',j_1'}^* = P_{p,q}^* > P_{p,q}^{\star}$ can occur. There are then three possible subcases:
\begin{enumerate}[a)]
\item If $(i'_3,j'_3) > (i_3,j_3) $, it is easy to see that there exists $(i^0,j^0)$ with $(i_2',j_2') \leq (i^0,j^0)  <(i_3,j_3) $ such that $P_{i^0,j^0}^* < P_{i^0,j^0}^{\star}$, since the EH constraint at the $(i_3,j_3)$-th communication block is achieved with equality for $\{ P_{i,j}^{\star} \}$. Then, it follows that $ P^*_{i'_3,j'_3-1} = P^*_{i^0,j^0} < P^{\star}_{i^0,j^0} = P^{\star}_{i_3,j_3-1} \leq P^{\star}_{i'_3,j'_3-1}$\footnote{$P_{n,m-1}$ denotes the power of the $((n,m)-1)$-th communication block, similar for $P_{n,m+1}$.}, where the first equality is due to the assumption that $\{ P^*_{i,j} \}$ keep identical among the $(i'_2,j_2')$-th to the $((i_3',j_3')-1)$-th communication blocks. Thus, by the obtained $ P^*_{i'_3,j'_3-1} < P^{\star}_{i'_3,j'_3-1}$, it is easy to see that for $\{ P^*_{i,j} \}$, the EH constraint at the $((i'_3,j'_3)-1)$-th communication block cannot be achieved with equality, which contradicts with the assumption. Therefore, this subcase cannot occur.

\item If $(i'_3,j'_3) = (i_3,j_3) $, it is obvious that $\{ P^*_{i,j} \}$ cannot be optimal due to the search procedure of Algorithm \ref{opt_power_allocation_alg} and Theorem \ref{opt_power_allocation}.

\item If $(i'_3,j'_3) < (i_3,j_3) $, we check the power profiles from the $(i'_3,j'_3)$-th to the $((i_3,j_3)-1)$-th communication blocks, and denote the index of the last power changing block within this region for power profile $\{ P_{i,j}^* \}$ as $(i^0,j^0)$. First, we show that $P_{i_3,j_3-1}^* \leq P_{i_3,j_3-1}^{\star}$. This can be proved by contradiction: Assume that $P_{i_3,j_3-1}^* > P_{i_3,j_3-1}^{\star}$, and by the non-decreasing property of $\{ P_{i,j}^{\star} \}$, it follows that $P_{i,j}^* = P_{i_3,j_3-1}^* > P_{i_3,j_3-1}^{\star} \geq P_{i,j}^{\star}$ for $(i^0,j^0) \leq (i,j) \leq ((i_3,j_3)-1)$, or in other words,
\begin{align} \label{contra1}
\sum_{(i^0,j^0) \leq (i,j) \leq ((i_3,j_3)-1)} P_{i,j}^* >  \sum_{(i^0,j^0) \leq (i,j) \leq ((i_3,j_3)-1)} P_{i,j}^{\star}.
\end{align}
Since at the $((i^0,j^0)-1)$-th communication block, the EH constraint is achieved with equality for $\{P_{i,j}^{*}\}$ and may not for $\{P_{i,j}^{\star}\}$, it is obtained that
\begin{align} \label{contra2}
\sum_{(i^0,j^0) \leq (i,j) \leq ((i_3,j_3)-1)} P_{i,j}^{*}  \leq  \sum_{(i^0,j^0) \leq (i,j) \leq ((i_3,j_3)-1)}  Q_i \leq \sum_{(i^0,j^0) \leq (i,j) \leq ((i_3,j_3)-1)} P_{i,j}^{\star},
\end{align}
which contradicts with (\ref{contra1}). Therefore, we conclude that $P_{i_3,j_3-1}^* \leq P_{i_3,j_3-1}^{\star}$.

Then, we claim that for $\{ P_{i,j}^* \}$, the EH constraint at the $((i_3,j_3)-1)$-th communication block is achieved with equality. We prove this result by contradiction. Assume that for $\{ P_{i,j}^* \}$, the EH constraint at the $((i_3,j_3)-1)$-th communication block is not achieved with equality. Thus, by Proposition \ref{prop_one_high}, it follows that $ P_{i_3,j_3-1}^* = P_{i_3,j_3}^* $. Moreover, since $P_{i_3,j_3}^{\star} \geq P_{i_3,j_3-1}^{\star} \geq P_{i_3,j_3-1}^* = P_{i_3,j_3}^*$, it follows that for $\{ P_{i,j}^* \}$, the EH constraint at the $(i_3,j_3)$-th communication block is not achieved with equality. Repeat the above argument, and we conclude that for $\{ P_{i,j}^* \}$, the EH constraint after the $((i_3,j_3)-1)$-th communication block is not achieved with equality, which means that $\{ P_{i,j}^* \}$ cannot be optimal and shows the contradiction.

Finally, we only focus on the $(1,1)$-th to the $((i_3,j_3)-1)$-th communication blocks. Since for both $\{ P_{i,j}^* \}$ and $\{ P_{i,j}^{\star} \}$, the EH constraint at the $((i_3,j_3)-1)$-th communication block is achieved with equality, we can apply the result in Theorem \ref{opt_power_allocation}, and show that $\{ P_{i,j}^* \}$ is not optimal.

\end{enumerate}

\item If $(p,q)=(i_1,j_1)$, we consider the following two sub-cases:

\begin{enumerate}[a)]
  \item $P_{p,q}^* > P_{p,q}^{\star}$: This case cannot be optimal due to similar arguments as those in Case \ref{proof_opt_case1}).

  \item $P_{p,q}^* < P_{p,q}^{\star}$: For this case, we consider the following subcases:
      \begin{enumerate}[i)]
        \item \label{opt_proof_EH_case_equ}$P_{p,q}^* > 0$: Compared to $\{P_{i,j}^{\star}\}$, the power profile $\{P_{i,j}^*\}$ saves some power at the $(p,q)$-th communication block. It is easy to see that the saved power with the amount $P_{p,q}^{\star} - P_{p,q}^*$ is allocated to the blocks after the one with index $(i_1,j_1)$ in a ``waterfilling'' manner with the baseline given by  $\{ P_{i,j}^{\star}\}$, $(i,j)> (i_1,j_1)$: First, we equally allocate $P_{p,q}^{\star} - P_{p,q}^*$ to the blocks with the indices from $(i_2,j_2)$ to $(i_3,j_3)-1$ (due to Proposition \ref{prop_one_high}), and still keep these blocks with identical power allocation. Only when the newly obtained $P_{i,j}^*$, $(i_2,j_2) \leq (i,j) \leq (i_3,j_3)-1$, exceeds the power values of the blocks with the indices from $(i_3,j_3)$ to $(i_4,j_4)-1$, the power values of the blocks from $(i_2,j_2)$ to $(i_4,j_4)-1$ should be maken identically, and so on. In other words, we claim that there exists $k$, $3 \leq k \leq n+1$, such that $(i_k,j_k) = (i'_3,j'_3)$. Thus, for $\{P_{i,j}^*\}$, the power allocation problem from the $(i'_1,j'_1)$-th to $((i'_3,j'_3)-1)$-th blocks corresponds to one single iteration in Step (\ref{opt_power_EH_search_p0}) of Algorithm \ref{opt_power_allocation_alg}, which is shown to be not optimal by the searching procedures.

        \item $P_{p,q}^* = 0$: We define a new power allocation problem starting from the $(i_2,j_2)$-th to the $(N,M)$-th communication blocks by setting $P_{i,j}=0$, $(i,j) < (i_2,j_2)$ for Problem (P1). It is easy to observe that: (A) for the optimal power profile of this newly defined problem, all power values are no smaller than $P_b$, since $P^{\star}_{i,j} \geq P_b$, $(i,j) \geq (i_2,j_2)$; and (B) there exists $k$, $3 \leq k \leq n+1$, such that the optimal power profile of the newly defined problem after the $(i_k,j_k)$-th communication block is the same as the corresponding part of $\{P^{\star}_{i,j}\}$ (due to similar argument as Case (\ref{opt_proof_EH_case_equ})). Thus, by (A), we conclude that the optimal value of the newly defined problem is no larger than the average outage probability given by $\{ P^*_{i,j}\}$ (note that the power allocation of the newly defined problem is operated on the convex region of its objective function and the arguments of Theorem 1 in \cite{yang} can be applied to show this result); by (B), considering the average outage probability minimization over the whole $NM$ communication blocks, we conclude that $\{ P^{\star}_{i,j} \}$ outperforms the power profile obtained by solving the newly defined problem (Since the newly defined power allocation problem is only a special case of the last iteration of Step (\ref{opt_power_EH_search_p0}) in Algorithm \ref{opt_power_allocation_alg} by choosing $P=0$ in (\ref{search_posi})). Thus, this case cannot be optimal.

      \end{enumerate}
\end{enumerate}

\item If $(p,q) > (i_1,j_1)  $, it follows that $P^*_{i,j} \geq P_b$ and $P^{\star}_{i,j} \geq P_b$ for $(i,j) \geq (p,q)$, which means that the power allocation is always in the convex regime of the objective function. By a similar argument as Theorem 1 in \cite{yang}, we know that this case cannot occur.

\end{enumerate}

By combining the above discussions, Theorem  \ref{deter_optimal_proof} is proved.



\begin{thebibliography}{99}
\bibitem{sharma}
V. Sharma, U. Mukherji, V. Joseph, and S. Gupta, ``Optimal energy management policies for energy harvesting sensor nodes,'' \emph{IEEE Trans. Wireless Commun.}, vol. 9, no. 4, pp. 1326-1336, Apr. 2010.

\bibitem{yang}
J. Yang and S. Ulukus, ``Optimal packet scheduling in an energy harvesting communication system,'' \emph{IEEE Trans. Commun.}, vol. 60, no. 1, pp. 220-230, Jan. 2012.

\bibitem{ulukus}
O. Ozel, K. Tutuncuoglu, J. Yang, S. Ulukus, and A. Yener, ``Transmission with energy harvesting nodes in fading wireless channels: optimal policies,'' \emph{IEEE J. Sel. Areas Commun.}, vol. 29, no. 8, pp. 1732-1743, Sep. 2011.

\bibitem{rui}
C. K. Ho and R. Zhang, ``Optimal energy allocation for wireless communications with energy harvesting constraints,'' \emph{IEEE Trans. Sig. Proc.}, vol. 60, no. 9, pp. 4808-4818, September, 2012.

\bibitem{chuan}
C. Huang, R. Zhang, and S. Cui, ``Throughput maximization for the Gaussian relay channel with energy harvesting constraints,'' \emph{to appear in IEEE J. Sel. Areas Commun.}. (Available online at arXiv:1109.0724)

\bibitem{ozel}
O. Ozel and S. Ulukus, ``Information-theoretic analysis of an energy harvesting communication system,'' \emph{International Workshop on Green Wireless (W-GREEN) at IEEE PIMRC}, Istanbul, Turkey, Sep. 2010.

\bibitem{rajesh}
R. Rajesh, V. Sharma, and P. Viswanath, ``Capacity of fading Gaussian channel with an energy harvesting sensor node,'' \emph{in Proc. IEEE Globecom}, Houston, US, Dec. 2011.

\bibitem{biglieri}
E. Biglieri, J. Proakis, and S. Shamai, ``Fading channels: information-theoretic and communications aspects,'' \emph{IEEE Trans. Inf. Theory},
vol. 44, no. 6, pp. 2619-2692, Oct. 1998.

\bibitem{goldsmith}
A. Goldsmith and P. P. Varaiya, ``Capacity of fading channels with channel side information,'' \emph{IEEE Trans. Inf. Theory}, vol. 43, no. 6,
pp. 1986-1992, Nov. 1997.


\bibitem{caire}
G. Caire, G. Taricco, and E. Biglieri, ``Optimal power control over fading channels,''
{\it IEEE Trans. Inf. Theory}, vol. 45, no. 5, pp. 1468-1489, Jul. 1999.

\bibitem{telater}
I. E. Telatar, ``Capacity of multi-antenna Gaussian channels,'' {\it Eur. Trans. Telecommun.}, vol. 10, no. 6, pp. 585-595, Nov. 1999.

\bibitem{shamai}
L. H. Ozarow, S. Shamai, and A. D. Wyner, ``Information theoretic considerations for cellular mobile radio,'' {\it IEEE Trans. Veh. Technol.}, vol. 43 no. 2, pp.
359-378, 1994.

\bibitem{luo}
S. Luo, R. Zhang, and T. J. Lim, ``Optimal save-then-transmit protocol for energy harvesting wireless transmitters,'' {\it IEEE Trans. Wireless Commun.}, vol. 12, no. 3, pp. 1196-1207, Mar. 2013.

\bibitem{xu}
J. Xu and R. Zhang, ``Throughput optimal policies for energy harvesting wireless transmitters with non-ideal circuit power,'' {\it to appear in IEEE J. Sel. Areas Commun.}. (available online at arXiv:1204.3818)


\bibitem{tse}
D. Tse and P. Viswanath, \emph{Fundamentals of Wireless Communication}. Cambridge, UK: Cambridge University Press, 2005.

\bibitem{sagias}
N. C. Sagias and G. K. Karagiannidis, ``Gaussian class multivariate Weibull distributions: theory and applications in fading channels,'' \emph{IEEE Trans. Inf. Theory}, vol. 51, no. 10, pp. 3608-3619, Oct. 2005.

\bibitem{boyd}
S. Boyd and L. Vandenberghe, \emph{Convex optimization}. Cambridge, UK: Cambridge University Press, 2004.

\bibitem{chuan_full}
C. Huang, R. Zhang, and S. Cui, ``Optimal Power Allocation for Outage Minimization in Fading Channels with Energy Harvesting Constraints'', Available online at arXiv:1212.0075.

\bibitem{bell}
D. P. Bertsekas, \emph{Dynamic Programming and Optimal Control Vol. 1}. Belmont, MA: Athena Scientific, 1995.

\end{thebibliography}
\end{document}